\documentclass[manuscript, nonacm]{acmart}

\setcopyright{none}
\copyrightyear{2021}
\acmYear{2021}
\acmDOI{}
\acmConference[]{}{}{}
\acmBooktitle{}
\acmPrice{}
\acmISBN{}

\usepackage{microtype}
\usepackage{latexsym}
\usepackage{amsmath}
\usepackage{amsthm}
\usepackage{enumitem}
\usepackage{graphicx}

\DeclareMathOperator{\poa}{PoA}
\DeclareMathOperator{\ibr}{IoBR}
\DeclareMathOperator{\opt}{OPT}
\DeclareMathOperator{\br}{BR}

\newcommand{\A}{\ensuremath{\mathcal A}}
\newcommand{\F}{\ensuremath{\mathcal F}}

\setcitestyle{acmnumeric}

\title[Classifying Convergence Complexity of Nash Equilibria in Graphical Games Using Distributed Computing Theory]{Classifying Convergence Complexity of Nash Equilibria in Graphical Games Using Distributed Computing Theory}

\author{Juho Hirvonen}
\affiliation{%
	\institution{Aalto University}
	\city{Espoo}
	\country{Finland}
}
\email{juho.hirvonen@aalto.fi}

\author{Laura Schmid}
\affiliation{%
	\institution{IST Austria}
	\city{Klosterneuburg}
	\country{Austria}
}
\email{laura.schmid@ist.ac.at}

\author{Krishnendu Chatterjee}
\affiliation{%
	\institution{IST Austria}
	\city{Klosterneuburg}
	\country{Austria}
}
\email{krishnendu.chatterjee@ist.ac.at}

\author{Stefan Schmid}
\affiliation{%
	\institution{University of Vienna}
	\city{Vienna}
	\country{Austria}
}
\email{stefan_schmid@univie.ac.at}

\begin{abstract}
	Graphical games are a useful framework for modeling the interactions of (selfish) agents who are connected via an underlying topology and whose behaviors influence each other. They have wide applications ranging from computer science to economics and biology. 
	Yet, even though a player's payoff only depends on the actions of their direct neighbors in graphical games, computing the Nash equilibria and making statements about the convergence time of ``natural'' local dynamics in particular can be highly challenging. In this work, we present a novel approach for classifying complexity of Nash equilibria in graphical games by establishing a connection to local graph algorithms, a subfield of distributed computing. In particular, we make the observation that the equilibria of graphical games are equivalent to locally verifiable labelings (LVL) in graphs; vertex labelings which are verifiable with a constant-round local algorithm. This connection allows us to derive novel lower bounds on the convergence time to equilibrium of best-response dynamics in graphical games. Since we establish that distributed convergence can sometimes be provably slow, we also introduce and give bounds on an intuitive notion of ``time-constrained'' inefficiency of best responses. We exemplify how our results can be used in the implementation of mechanisms that ensure convergence of best responses to a Nash equilibrium. Our results thus also give insight into the convergence of strategy-proof algorithms for graphical games, which is still not well understood.
\end{abstract}

\begin{document}

\maketitle

\section{Introduction}

Modeling the interactions of multiple selfish agents, whose decisions and behavior influence each other and are in some way dependent on an underlying topology, is an important aspect of solving problems from a wide range of diverse fields. Game theory offers a natural approach to this issue in the form of multiplayer network games~\cite{jackson2015games,bramoulle07public}. Be it resource allocation~\cite{Avni2016}, routing~\cite{roughgarden2007routing}, provision of public goods~\cite{jackson2015games}, virus inoculation~\cite{Aspnes2006}, voting~\cite{Stewart2019}, or the spread of trends through social networks~\cite{abramson2001social}, instances of network games are everywhere to be found. They all have in common that interactions and players' decisions are in some way governed by the underlying graph and that players' payoffs usually depend on the network that links them. 

Previous literature suggests that general network games, such as the rich class of congestion games~\cite{roughgarden2007routing}, can be difficult to analyze, especially if treating them as a single monolithic group. Considering their heterogeneity and complexity, a considerable amount of research has thus been dedicated to more restricted versions of such games. One approach is to consider simplifications of specific network games such as singleton congestion games~\cite{bhawalkar2014weighted}. 
A different and less limiting approach is to focus on ``succinctly representable'' network games: each player's utility is solely determined by their own and their direct neighbors' actions. These so-called graphical games, introduced by Kearns et al.~\cite{kearns01graphical} capture the locality of effects on players, and are useful for settings where there are only a few strong direct influences on every agent, i.e.~when the underlying graph has a low degree~\cite{bramoulle07public}. These games have strong connections to Bayesian probabilistic inference networks (``probabilistic graphical models''). 

Graphical games have been subject to a wide range of research in their own right, see~\cite{Kearns2007survey} for a comprehensive survey. Computing the Nash equilibria of graphical games or proving results about their properties or convergence is however challenging in general. Previous work has dealt with aspects such as deciding whether a given strategy profile is a Nash equilibrium or whether equilibria exist, and classifying the complexity of these tasks~\cite{Schoenebeck2012}. The (approximate) computation of equilibria, often for various specific types of graphs, has also been a focus of a considerable number of studies. Previous work uses a variety of approaches for this. Kearns et al.~give a polynomial dynamic programming algorithm for computing approximate equilibria in trees, inspired by belief propagation. This algorithm can also be extended to a distributed message-passing scheme for more general topologies~\cite{Ortiz2003}; however, it is important to note that the result is not guaranteed to be efficient, nor is it strategy-proof: Indeed, players can have incentives to deviate from the algorithm if they are actually attempting to reach the Nash equilibrium they compute, as the computation does not reflect natural and rational game dynamics. Follow-up work gave some intractability results for this algorithm when it is used to calculate exact equilibria~\cite{Elkind2006}.

Other works try to generalize beyond specific graph structures: Daskalakis and Papadimitriou~\cite{Daskalakis2006} relate graphical games to random Markov nets. This reduction can then be used in the computation of pure Nash equilibria.
The authors further show polynomial complexity of finding pure equilibria on graphs with bounded treewidth, and give exponential algorithms for computing approximate mixed equilibria. Other papers, such as~\cite{Vickrey2002}, use multi-agent algorithms like hill-climbing or constraint satisfaction approaches to calculate approximate equilibria. These heuristics show good performance, but have no worst-case guarantees. Correlated equilibria in graphical games and related intractability results have also been of interest, for example in~\cite{papadimitriou2008computing}. Furthermore, Jackson and Yariv~\cite{jackson2007diffusion} investigate best-response dynamics and diffusion of behavior in a dynamic form of graphical games. They present comparative statics results and investigate (Bayesian) equilibrium stability when behaviors can propagate.

Despite this rich and multifaceted literature, we still lack a systematic understanding and classification of graphical games in terms of (distributed) computational complexity; in particular, the convergence time of strategy-proof algorithms and local dynamics such as best-response to Nash equilibria is still not properly understood, except for special cases such as~\cite{Komarovsky15}. 

This paper presents a novel approach for shedding light on the distributed complexity of Nash equilibria and the dynamics of the convergence behavior in graphical games, using the perspective of distributed computing. In particular, we establish a connection to distributed local graph algorithms~\cite{Peleg2000}, which is not only natural and intuitive, but also allows us to leverage some of the analytical techniques and powerful results developed in this field over the last decades. More specifically, we show that the equilibria of graphical games are equivalent to locally verifiable labelings (LVL) in graphs~\cite{korman2010proof}: LVLs are solutions to distributed graph problems equivalent to the game. In a nutshell, a labelling is locally verifiable if and only if there exists a constant-round distributed algorithm that can check if the labelling is correct around each node.
For verification, each node thus simply checks the strategy that is assigned to itself and its neighbors, and can hence locally check if this action is in equilibrium. 

This close connection has only been rarely explored, for example in work that considers extensive game formulations of distributed algorithms~\cite{collet2018equilibria}. We leverage it to derive several more general results on graphical games.
First, we prove that best response dynamics can converge only as fast as a distributed algorithm can compute 
solutions to an equivalent graph problem. With this, we can classify the convergence properties and efficiency of \emph{local} dynamics in network games by classifying the corresponding LVL problem.

Since we observe an inherently slow distributed convergence time in many scenarios, 
we then give a more fine-grained view on the costs of distributed computations, and introduce a new notion of a time-constrained inefficiency of best response dynamics, a natural measure of efficiency when convergence to Nash equilibria can be slow. We exemplify our results with two simple and fundamental graphical games that have been well studied in the literature: the one-shot public goods game~\cite{bramoulle07public} and a simple minority game~\cite{challet1997emergence} (also called ``social game'' or ``restaurant game'' in previous literature, e.g.~\cite{Bramouille07social} ). For these games, we present convergence properties, lower bounds, and efficiency results, which all serve to highlight the key novelty of our approach.

We further show that in cases where best responses do not converge to the equilibrium despite the provable existence of an efficiently computable one, our techniques can be used to point to mechanism implementations where this can be guaranteed. Since best-response dynamics are rational with respect to the restricted local knowledge of nodes, our work also connects to the open question of strategy-proof algorithms for distributed Nash equilibria computation that is posited in~\cite{Kearns2007survey}.

\section{Model}
In the following, we provide the necessary preliminaries for our results by giving definitions of game-theoretic concepts and revisiting important cornerstones of distributed complexity theory and distributed graph problems.

\subsection{Graphical games}
 
We first define a multiplayer game in its general form as consisting of $n$ players $i \in \mathcal{I}$, each equipped with a pure strategy space or action space $A_i$. The cartesian product $A_1 \times A_2 \times ...\times A_n$ of the action spaces of individual agents is denoted by $\mathcal{A}$. Furthermore, every player has a utility function $u_i(\vec{a})$ for each profile of strategies or actions $\vec{a}=(a_1,a_2,...a_n)$, i.e. a mapping $u_i \colon \mathcal{A} \times N \to \mathbb{R}$. We use $a_{-i}$ to denote the joint strategy profile of all players except for player $i$.

Throughout the paper, we will focus on pure strategies, and later give an outlook on mixed strategy extensions in the conclusion section. A mixed strategy $\sigma$ is a probability distribution over pure strategies or actions, with the corresponding strategy space $\Sigma_i$ for each player. We can then define joint mixed strategy profiles as elements of the product space $\Sigma= \times_i \Sigma_i$.

We can now define the central concept of (pure) \emph{Nash equilibrium}. A strategy or action profile $\vec{a^*}$ is a Nash equilibrium if for all players $i$ it holds that 

\begin{equation}
u_i(a^*_i,a^*_{-i}) \geq u_i(a_i,a^*_{-i}) ~~ \forall a_i \in A_i
\end{equation}

In other words, no player in a Nash equilibrium can gain a higher utility by unilaterally deviating. We say that the strategy $a^*_i$ is a \emph{best response} to the rest of the strategy profile $a^*_{-i}$. 

In the case of mixed strategies, a Nash equilibrium is characterized by the relation

\begin{equation}
u_i(\sigma^*_i,\sigma^*_{-i}) \geq u_i(a_i,\sigma^*_{-i}) ~~ \forall a_i \in A_i,
\end{equation}

\noindent i.e. it suffices to check pure strategy deviations to confirm that a mixed strategy profile is an equilibrium.

We can now define \emph{graphical games}, given by a triple $(\mathcal{A}, u, N)$. As introduced in~\cite{kearns01graphical}, graphical games are a concisely representable form of \emph{multiplayer games} on networks.
A \emph{network} or a \emph{graph} $N = (V,E)$ consists of a set of $n$ \emph{nodes} $V$ and a set $E$ of \emph{edges} between pairs of nodes. The nodes of the network $N$ then represent the agents, or players, in the graphical game. We define the local neighborhood of a node $v$ as $B(v) \subseteq \{1,..,n\} = \{j \in V,\, (v,j) \in E\}$, with $v \in B(v)$ as well. As in normal form multiplayer games, each agent $v$ is equipped with an action space $A_v$, and $\mathcal{A}$ denotes the product of these action spaces. A player's \emph{utility function} $u_v(\vec{a}^v)$ now depends only on a strategy profile restricted to their local neighborhood $B(v)$, i.e. the partial strategy profile $\vec{a}^v$. We denote the product of the individual agents' utility functions by $u$.

Throughout this work, we will assume that the network $N$ has a constant maximum degree $\Delta$. 

In this work, we consider the question of \emph{convergence} to Nash equilibria via \emph{best-response dynamics}, a specific example of \emph{local dynamics}. We assume that players can update their strategy between rounds of the game. To this end, they use the rule to update their strategy with the best response to their neighborhood's strategy profile. That is, the action of a node $v$ in round $t$ is a best response to the partial strategy profile $\vec{a}^v(t-1)$. There is some fixed order on the actions that is used to break ties. We assume that nodes only have local information and cannot look beyond their neighborhood. Their restricted knowledge makes such local dynamics rational. In this work, we will subsequently only consider best responses. However, our results also hold for more general local dynamics.

In order to have a reasonable definition of a running time for local dynamics, we define a model of \emph{fair best responses}. The play consists of \emph{fair rounds}. During each round the adversary schedules all agents to act exactly once and one at a time. The convergence time of best-response dynamics on a fixed network is the maximum number of fair rounds until all players have reached a Nash equilibrium over, all possible orders of play. For random initial strategy profiles we say that best-response dynamics converge if they converge with high probability over the initial strategy assignment.

\subsection{Distributed complexity theory and locally verifiable labelings}

We next present some preliminaries on distributed graph algorithms and complexity.
In particular, this paper will establish a connection of graphical games to the
\emph{LOCAL} model of computation~\cite{Linial1992,Peleg2000}: in this model, we are given
a fixed network $N=(V,E)$ (a graph) connecting $n=|V(N)|$ nodes. The nodes
collaboratively aim to solve a given task, but can only communicate with their neighbors 
in the graph. The computation proceeds in \emph{synchronous} rounds, and in each round each node in the graph can send a message to each of its neighbors
(there is no bound on the message size), receive a message from each neighbor, and perform
local computations (there is no bound on the complexity of these local computations). Initially the nodes do not know anything about the input network. Each node is responsible for computing its own part of the output. 
The \emph{(distributed) complexity} of a local algorithm is measured in the number of communication rounds
until all nodes have computed their outputs.

To identify the nodes, each node gets an $O(\log n)$-bit unique name as an input. We will also consider the randomized LOCAL model, where instead each node has access to its own private source of random bits. The randomized LOCAL model is at least as strong as the deterministic LOCAL model, as randomness can be used to generate unique identifiers with high probability.

The crucial property of the LOCAL model is that in $t$ communication rounds each node can gather exactly its $t$-hop neighborhood in the network. Information that is outside this radius cannot affect the actions of the node, since it hasn't had time to travel to it.  Since communication is not limited, the nodes can gather \emph{all} information about their $T$-hop neighborhood in $T$ rounds. This implies that distributed algorithm with complexity $T$, without loss of generality, can be thought of as function that maps the input-labelled $T$-hop neighborhoods to the possible outputs.

The $T$-neighborhood of a node $v$ is denoted by
$B(v,T)$. The radius-1 neighborhood is denoted simply by $B(v)$.

We are particularly interested in a class of LOCAL problems 
called \emph{locally verifiable labelling~(LVL)} problems.
An LVL consists of an \emph{alphabet} $\Sigma$ and a set of \emph{configurations} $\mathcal{C}$. The alphabet $\Sigma$ is simply some possibly infinite set of labels. Each configuration $C \in \mathcal{C}$ is a subgraph
of $N$ centered on some node $v$ with radius at most $k$ for some constant $k$ (the verification radius). Each node of $C$ is labelled with some element $\sigma \in \Sigma$. In this work we only consider LVLs with radius $k = 1$. 
A mapping $f\colon V \to \Sigma$ of the graph $N$ is a solution to $P$ if and only if each $f$-labelled $k$-neighborhood of $N$ is a configuration in $\mathcal{C}$. 
Locally verifiable labellings are a generalization of \emph{locally checkable labellings (LCLs)}~\cite{Naor1995}, a family of problems that has been studied extensively in recent years (for example \cite{Brandt2016,Chang19exponential,Chang19time,Balliu2018stoc,chang18complexity,brandt19automatic,balliu19binary}).

An important result in LOCAL algorithms theory, which will also be relevant for our work, is related to a ``complexity gap'': on general bounded-degree graphs, i.e., if the maximum degree of the graph $N$ is bounded by a constant $\Delta$,
the deterministic distributed complexity of an LVL is either $O(\log^* n)$, 
or $\Omega(\log n)$ and the randomized complexity is either $O(\log^* n)$ or $\Omega(\log \log n)$~\cite{Chang19exponential}. 
Here, $\log^* n$ is the \emph{iterated logarithm} (pronounced ``log-star''), 
a function which grows significantly more slowly
than the logarithm: e.g.\ $\log^*$ of the number of atoms in the observable universe is 5. Formally, $\log^* n$ is defined as:
$$
\forall x\leq 2: \log^* x:=1, ~~ \forall x>2: \log^* x:= 1 + \log^*(\log x)
$$
The complexity gap on bounded-degree networks is also the best possible: there exists an LVL such that its deterministic complexity is $\Theta(\log n)$ and the randomized complexity is $\Theta(\log \log n)$~\cite{Chang19exponential,Brandt2016,ghaffari17distributed}. This also proves that the deterministic and randomised complexities of an LVL can be exponentially separated.

On other graph families, the complexity gap can be even larger.
For example, on paths and cycles LVLs have complexity either $O(\log^* n)$ or $\Theta(n)$~\cite{Chang19exponential}, and on grids and toruses on $n$ nodes, either $O(\log^* n)$ or $\Omega(\sqrt{n})$ \cite{Brandt2017}.

Throughout this paper, we will call a distributed algorithm with complexity $O(\log^* n)$ \emph{efficient}.

\section{Distributed computing and graphical games}

Our work is motivated by our observation that graphical games and local algorithms are fundamentally connected. In particular, we observe that all Nash equilibria (and in fact all equilibria that are based on local information only) are LVLs. If we assume that agents playing a game converge to an equilibrium they are implicitly solving the corresponding computational task.

The LOCAL model has two particular properties. First, it is possible to prove \emph{unconditional} impossibility results in the LOCAL model. Existing results cover many LVLs that are potentially interesting from the perspective of game theoretical applications~(see e.g.\  \cite{balliu19maximal,Chang19exponential,balliu19weak,Brandt2016,Brandt2017,chang18complexity,Linial1992}). 
Second, it is a strong distributed model in the sense that algorithms in the LOCAL model are only limited by information propagation, and not e.g.\ by failures or bandwidth limitations. This means that any \emph{impossibility results} proven in the LOCAL model apply to a wide range of more realistic models. In particular, they apply to any model of games where the play of the agents is constrained by the available information.

An interesting property of LVLs in the LOCAL model is the complexity gap: they can either be computed efficiently in $O(\log^* n)$ rounds, or require $\Omega(\log n)$ rounds in the deterministic LOCAL model and $\Omega(\log \log n)$ rounds in the randomized LOCAL model \cite{Chang19exponential}. If we can show that the Nash equilibria of a game are not efficiently solvable then this implies that in that game cannot best responses cannot converge fast. 

To formalize this connection, we must transfer existing impossibility results for computational tasks in the distributed setting to a model of games. To do this, we first establish an equivalence between the Nash equilibria of graphical games and locally verifiable labellings. Then, to transfer impossibility results to our model of fair sequential best responses, we show that the LOCAL model can simulate best responses. This implies that if all Nash equilibria of a game are hard to compute as LVLs, then the best-response dynamics cannot converge to those equilibria fast.

\begin{theorem} \label{thm:ne-are-lvls}
	Let $G = (\A, u, N)$ be a graphical game. The Nash equilibria of $G$ uniquely define a locally verifiable labelling $P$.
\end{theorem}

\begin{proof}
	This follows from the locality of the utility functions of graphical games.
	Let $G = (\A, u, N)$ be a game on a network, and let $\vec{a}$ be some strategy profile of $G$. Observe that $\vec{a}$ is a Nash equilibrium if and only if for each agent $v \in V(N)$ their current strategy $a_v$ maximizes their utility over all choices, and this only depends on the strategies $a_u$ for all neighbors $u \in B(v)$.
	That is, we can define the set of Nash equilibria of a game $G$ as an LVL $P(G)$ as follows: let $\mathcal{S}$ consist of every radius-1 subgraph $S$ of $N$. Let $\mathcal{C}$ consist of every copy of each $S \in \mathcal{S}$ labelled with the actions of $G$ such that the strategy of the center node is a best response to the strategies of its neighbors with respect to $u$. We have that the alphabet $\Sigma$ consists of all possible actions of $\A$.
\end{proof}

It should be stressed that the correspondence does not use any properties specific to pure Nash equilibria. For example mixed Nash equilibria of graphical games also define LVLs.

Next we show how best responses can be simulated in the LOCAL model. Let $G = (\A, u, N)$ be a graphical game. We construct a corresponding instance in the LOCAL model by taking the network $N$. Then, if we consider the deterministic LOCAL model, an adversary assigns $O(\log n)$-bit names to the nodes. In the randomized LOCAL model, each node gets a uniformly random infinite string of bits as input instead. We show that if the best responses converge in $T(n)$ rounds, then this can be turned into a distributed algorithm that computes the corresponding Nash equilibrium in $O(\log^* n + T(n))$ rounds.

\begin{theorem}[best responses correspondence] \label{thm:local-dynamics-local}
	Fix a function $T$. Consider a graphical game $G$.
	\begin{enumerate}[noitemsep]
		\item If the best-response dynamics converge on $G$ in $T(n)$ rounds from a constant initial strategy profile, then there exists a deterministic distributed algorithm in the LOCAL model that solves the LVL corresponding to the Nash equilibria of each $G$ in $O(\log^* n + T(n))$ rounds.
		\item If the best-response dynamics converge with high probability on $G$ in $T(n)$ rounds from a random initial strategy profile, then there exists a randomized distributed algorithm in the LOCAL model that solves the LVL corresponding to the Nash equilibria of $G$ in $O(\log^* n + T(n))$ rounds with high probability.
	\end{enumerate}
\end{theorem}

The proof of Theorem~\ref{thm:local-dynamics-local} is given in Appendix~\ref{app:local-dynamics-proof}. The theorem implies that the convergence time of best-responses is bounded by impossibility results from distributed computing.

\begin{corollary} \label{cor:local-dynamics-lb} 
	Assume that the Nash equilibria of a graphical game $G$, as LVLs, have deterministic complexity $\Omega(T(n))$ and randomized complexity $\Omega(T'(n))$, for any $T(n), T'(n) = \Omega(\log^* n)$. Then best-response dynamics for $G$ require $\Omega(T(n))$ and $\Omega(T'(n))$ rounds to converge from a constant and a randomized initial strategy profile, respectively.
\end{corollary}

\begin{proof}
	This follows from Theorem~\ref{thm:local-dynamics-local}: if best-response dynamics converge faster, then this can be turned into a fast distributed algorithm, a contradiction.
\end{proof}

The complexity gap of LVLs in the LOCAL model, in the context of our result, implies that if none of the Nash equilibria of a game are efficiently computable, then best responses converge significantly slower.

\begin{corollary} \label{cor:no-fast-convergence}
	Let $G$ be a game such that none of its Nash equilibria can be solved in time $O(\log^* n)$ as LVLs. Then the best-response dynamics require $\Omega(\log n)$ rounds to converge from constant initial strategy profile, and $\Omega(\log \log n)$ rounds to converge from a random initial strategy profile.
\end{corollary}

In the next sections we look at examples of graphical games, their Nash equilibria, and the corresponding LVL problems. We illustrate the correspondence in Figure~\ref{fig:fig1}.

\begin{figure}[ht]
\centering
\includegraphics[width=1\textwidth]{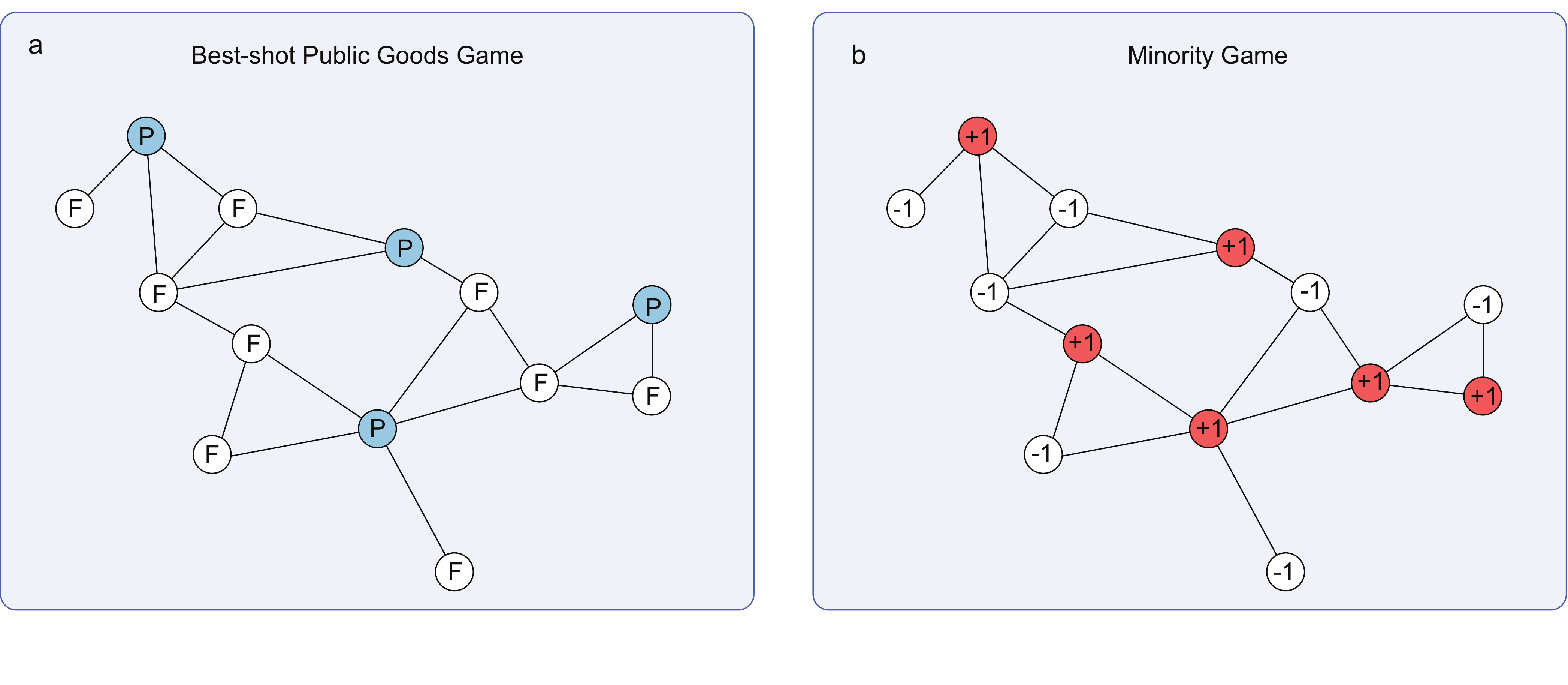}
\caption{\label{fig:fig1} We illustrate the correspondence between Nash equilibria of graphical games and LVLs for the games described in Sections 3.1 and 3.2. \textbf{a, }In the best-shot public goods game, nodes can choose to produce a good (P) or to forego producing (F). The Nash equilibrium of this game is a solution to the maximal independent set problem. A node in this state can not do any better given that the others play their equilibrium strategy, and they can check this by simply by looking at their neighbors. This also means that they can check whether their computed solution to the maximal independent set is correct. \textbf{b, } The minority game gives players a choice between two actions, here -1 and +1. Their goal is to choose the opposite of the majority of their neighbors. The Nash equilibrium is a locally optimal cut: players have at least as many neighbors playing the opposite strategy as the same strategy. This is also a LVL: every agent can again locally check that they are in equilibrium, i.e. the correctness of the computed solution.}
\end{figure}

\subsection{Best-shot public goods game}

This game deals with the provision of public goods such as finding a cure for a disease or filling an important supply, assuming that only the maximal contribution counts towards the provision level instead of the sum of all players' contributions~\cite{cherry2013heterogeneity,bramoulle07public}. In contrast to public goods games often studied as social dilemmas, i.e. where the social optimum is reached by all players cooperating (producing) despite the incentive to do nothing, the problem here is not only one of free-riding, but also one of coordination. Player groups have to figure out which of them should optimally be the one to provide the good, in order to avoid redundant costs. Here we consider the simplest version of the game where players only have two choices, to provide a good, or not to provide it. More formally, each agent has two possible actions, i.e. $A_i=\{P,F\}$. Players' utilities $u_i$ are as following: if a focal agent plays $F$ and one of their neighbors plays $P$, the utility is $u_i=1$. If the agent plays $P$, $u_i=1-c$, but if they and all their neighbors play $F$, the utility is $0$. 
Simply put, providing the public good is costly, and it is preferable for a player to have a neighbor do so; however, they are still better off providing it themselves than if nobody in their neighborhood does so.

The correspondence of this game with distributed graph problems is a prominent one: As shown in~\cite{bramoulle07public}, the Nash equilibria of the best-shot public goods game correspond to \emph{maximal independent sets} of agents playing $P$. The set is independent (i.e.\ no two agents with $P$ are adjacent), as two adjacent agents with the strategy $P$ would have incentive to choose $F$. On the other hand the set is maximal (i.e. each agent plays $P$ or has a neighbor that plays $P$) as otherwise such an agent would have incentive to play $P$.     

Maximal independent set is an efficiently solvable LVL: it can be computed in $O(\log^* n)$ rounds~\cite{panconesi01simple}, and this is known to be the best possible complexity~\cite{Linial1992}. Correspondingly, best responses converge in two fair rounds for the best-shot public goods game. We can compare this with the complexity analysis of best-shot public goods games in~\cite{Komarovsky15}, and point out that our approach considers a natural concept of \emph{distributed} complexity whereas previous work usually takes a different, more centralized view.

\begin{theorem} \label{thm:best-shot-pgg-converge}
	Fair best responses in the best-shot public goods game converge in two rounds from any initial configuration.
\end{theorem}

\begin{proof}
	Assume that the system starts with some arbitrary strategy profile. We claim that after the first round the set of agents playing $P$ is independent and after the second round it is maximal.
	
	Assume that after the first round neighboring agents $u$ and $v$ play $P$. Then the one that played last would have seen that the other plays $P$, and their best response is $F$. Therefore the set of agents playing $P$ is independent. In the second round no agent will switch from $P$ to $F$, as all of their neighbors play $F$. If an agent plays $F$ and their neighbors play $F$ it will switch to $P$. After two rounds the agents playing $P$ form a maximal independent set, which is a Nash equilibrium.
\end{proof}

\subsection{Minority game}

In this elementary anti-coordination game (also called social game)~\cite{challet1997emergence,Bramouille07social}, players attempt to do the opposite of what their neighbors are doing. That is, they attempt to anti-coordinate with what the majority of their surrounding co-players do. For example, this could describe a situation where agents try to choose a restaurant to go to that's not overly crowded. We formalize this with a game where players again have two possible actions, i.e. $A_i=\{-1,1\}$. A focal player's utility $u_i$ is defined as $1+|\{ j\in V_N(i): a_j \neq a_i \}| - |\{ j\in V_N(i) : a_j = a_i \}|$, i.e the difference between the number of neighbors with a different label and the same label, plus 1 (to avoid Nash equilibria with utility 0).

The Nash equilibria correspond to strategy profiles where each agent has at least as many neighbors playing the opposite strategy as the same strategy. In this game, we again have a correspondence with a prominent graph problem: in distributed computing, the corresponding LVL is known as locally optimal cut. It is known to be a hard problem~\cite{balliu19weak}: on 3-regular graphs it requires $\Omega(\log n)$ deterministic time and $\Omega(\log \log n)$ randomized time.

\begin{theorem} \label{thm:sg-convergence-lb}
	The convergence time of best-responses for the minority game is $\Omega(\log n)$ from a constant initial state and $\Omega(\log \log n)$ from a random initial state.
\end{theorem}

\begin{proof}
	Balliu et al.\ \cite{balliu19weak} have shown that finding a locally optimal cut requires $\Omega(\log n)$ deterministic time and $\Omega(\log \log n)$ randomized time in the LOCAL model. This, together with Corollary~\ref{cor:local-dynamics-lb} implies the theorem.
\end{proof}

The impossibility result of Balliu et al.\ holds even if the algorithm is promised that the network is a 3-regular tree or a 3-regular graph of high girth. Therefore the result also applies to best responses in these graph families.

\section{Inefficiency of best responses} \label{sec:br-inefficiency}

In this section we study the efficiency of best responses with respect to computational constraints. Traditionally, for example in the context of \emph{Price of Anarchy}~\cite{KoutsoupiasP99}, one compares the total welfare under the \emph{best strategy profile} to the total welfare under the \emph{worst Nash equilibrium}. However, it might be that the best solution is hard to compute in a distributed fashion. In fact, it might be that the worst Nash equilibrium is also hard to compute.

We will define a notion of computational inefficiency of best responses with respect to a time bound $T$. We show that there exist games such that we can bound the inefficiency of best responses away from the Price of Anarchy. This illustrates that Price of Anarchy does not always fairly reflect the quality of solutions computed by best responses when time constraints are taken into account.

As we showed in Theorem~\ref{thm:sg-convergence-lb}, best responses provably might not converge efficiently to an equilibrium. The technique we present in this section allows us to study the evolution of total welfare produced by best responses \emph{even before convergence}. We do this by bounding the total welfare produced by \emph{any fast distributed algorithm}.

To measure the performance of best responses, we compare the total welfare produced by $T$ fair rounds of best responses to the best solution that can be computed in $T$ rounds in the randomized LOCAL model. We assume that the system starts from a random initial strategy profile. For a fixed game $G = (\mathcal{A}, u, N)$, let $\opt(N,T)$ denote the best solution, in terms of total welfare, that can be computed on $N$ in the LOCAL model in $T$. Let $\br(N,T)$ denote the random variable that represents the solution computed on $N$ by $T$ rounds of best responses starting from a random initial strategy. We consider random initial strategy profiles, as a constant (or worst-case) strategy profile can guide best responses to perform poorly, depending on the game.

For a game $G = (\mathcal{A}, u, N)$ and a time bound $T$, we define the $T$-inefficiency of best responses as
\[
	\ibr(G,T) = \frac{u(\opt(N,T))}{E[u(\br(N,T))]}.
\]
To estimate the quantity $\ibr(G,T)$ we will bound $\opt(N,T)$ from above using computational arguments, and bound $\br(N,T)$ from below using both arguments about the behavior of best-response dynamics \emph{and} computational arguments.

Note that when we consider the best strategy profile that is computable in $T$ communication rounds, we consider distributed algorithms for \emph{optimization} problems. That is, the algorithms might not compute a solution corresponding to some Nash equilibrium, but more generally any strategy assignment that tries to optimize the total welfare of all agents.

In the next two sections we show how the inefficiency of best responses can be estimated using tools from distributed computing. 

\subsection{Best-shot public goods game}

We begin by analysing the inefficiency of best responses in the best-shot public goods game. We show that Price of Anarchy can be bounded away from the inefficiency of best responses.

We prove that there exists an infinite family of best-shot public goods game instances such that even though good and bad solutions exist, no distributed algorithm can compute them efficiently. Therefore best responses cannot produce these solutions either.

To construct these instances, we argue that there exist graphs which have good and bad solutions and graphs that look locally the same to the first class of graphs, but do not have any good or bad solutions. We can then argue using standard indistinguishability arguments from distributed computing that fast algorithms perform poorly.

The following theorem states the outcome of our analysis for the best-shot public goods game.

\begin{theorem} \label{thm:bspgg-poa}
	Fix a function $T = o(\log_d n)$. For every $d \geq 3$ and sufficiently large $n_0$, there exists an instance $G = (\mathcal{A},u,N)$ of the best-shot public goods game such that $N$ is a $d$-regular network of size $n \geq n_0$ and the following hold.
	\begin{enumerate}[noitemsep]
		\item $\poa(G) = \frac{1-c/(d+1)}{1-c/2}$.
		\item $\ibr(G,T) \leq \frac{d - c\ln d}{d - c(2+\varepsilon)\ln d}$ for any $\varepsilon > 0$.
	\end{enumerate}
\end{theorem}

According to Theorem~\ref{thm:best-shot-pgg-converge}, best responses converge in two fair rounds. Therefore to analyze the $T$-inefficiency of best responses, we can bound the total welfare of the \emph{worst} Nash equilibrium that can be computed in $T$ communication rounds.

In the best-shot public goods game the total welfare is maximized by minimizing the number of producing agents while ensuring that each non-producing agent is adjacent to a producing agent. Such sets are known as \emph{dominating sets}. Not all dominating sets are maximal independent sets, but all maximal independent sets are dominating sets. Thus we can bound the size of the minimum maximal independent set by the size of the minimum dominating set.

The following lemma bounds the best solutions that distributed algorithms can compute on certain networks that have good solutions.

\begin{lemma} \label{lem:best-pgg-alg}	
	There is no randomized algorithm in the LOCAL model that finds an independent set of size $> ((2+\varepsilon)\ln d / d)n$ or a dominating set of size $< ((1+\varepsilon)\ln d / d)n$ in expectation in $o(\log_d n)$ rounds on the networks from Lemma~\ref{lem:good-graph-pgg}.
\end{lemma}

Lemma~\ref{lem:good-graph-pgg} is stated in Appendix~\ref{lem:good-graph-pgg}. 
The proof of Lemma~\ref{lem:best-pgg-alg} is a standard indistinguishability argument from distributed computing. There exist regular high-girth networks with good and bad solutions (Lemma~\ref{lem:good-graph-pgg}), and regular high-girth networks with no good or bad solutions (Lemma~\ref{lem:bad-graph-pgg}). Since networks are locally indistinguishable (i.e.\ from the perspective of any node, look the same up to any distance $T(n) = o(\log_d n)$), any distributed algorithm must behave the same way in both networks. Since the size of the solutions in the latter network is bounded, solutions that are larger or smaller, respectively, cannot be found in the network that does have such solutions. The proof is given in Appendix~\ref{app:graph-proofs}.

Using Lemma~\ref{lem:best-pgg-alg} we can prove Theorem~\ref{thm:bspgg-poa}.

\begin{proof}[Proof of Theorem~\ref{thm:bspgg-poa}]
	We bound the Price of Anarchy, the best $T$-time computable solution, and the performance of best responses on the networks given by Lemma~\ref{lem:good-graph-pgg}.
	\begin{enumerate}[noitemsep]
		\item Best solution and worst equilibrium. The best solution (also a Nash equilibrium) on $N$ contains a $1/(d+1)$-fraction of nodes in the producing set, giving a total welfare equal to $(1-c/(d+1))n$. The bipartition gives the \emph{worst} Nash equilibrium: half of the nodes are in the independent set, giving a total welfare of $(1-c/2)n$. The Price of Anarchy is thus exactly
		\[
			\frac{1-c/(d+1)}{1-c/2}.
		\]
		\item Best $T(n)$-time computable solution. The best solution that can be computed in $T(n) = o(\log_d n)$ rounds on $N$, by Lemma~\ref{lem:best-pgg-alg}, has at least $((1+\varepsilon)\ln d / d)n$ nodes in it (as it corresponds to a dominating set). This gives a total welfare of 
		\[
			\biggl(1-\frac{c(1+\varepsilon) \ln d}{d}\biggr)n.
		\]
		\item Total welfare of best responses. Since best responses converge to an equilibrium and a distributed algorithm can simulate best responses (Lemma~\ref{lem:simulation-dynamics}), the worst solution that best responses could compute, in expectation, is bounded by the largest maximal independent set a distributed algorithm can compute in $O(T(n))$ rounds. Since maximal independents are a subset of independent sets, by Lemma~\ref{lem:best-pgg-alg} the worst solution best responses can compute in expectation has at most $(2+\varepsilon)(\ln d / d)n$ nodes in the producing set. The total welfare is at least
		\[
			\biggl(1-\frac{c(2+\varepsilon) \ln d}{d}\biggr)n.
		\]
	\end{enumerate} 
	From the last two we get that on $N$, the $T$-inefficiency of best responses is at most
	\[
		\frac{d - c\ln d}{d - c(2+\varepsilon)\ln d}
	\]
	for any $T = o(\log_d n)$.
\end{proof}

\subsection{Minority game}

Similar to the previous section, we can show that there exist instances of the minority game on which best responses perform relatively better than Price of Anarchy would indicate.

The proof again uses uses an indistinguishability argument to show that on certain networks fast distributed algorithm cannot find good solutions even though they do exist. In addition, we analyse best responses in the minority game and note that they only improve the total welfare. We prove the following theorem.

\begin{theorem} \label{thm:hd-poa}
	Fix a function $T = o(\log_d n)$. For every even $d \geq 4$ and large enough $n_0$, there exists a $d$-regular instance of the minority game on $n \geq n_0$ nodes such that the following hold.
	\begin{enumerate}[noitemsep]
		\item $\poa(G) = 2(d+1)$.
		\item $\ibr(G,T) \leq 1 + 2d/\sqrt{d-1}$.
	\end{enumerate}
\end{theorem}

To prove Theorem~\ref{thm:hd-poa} we construct two networks such that both look locally the same but one has a large cut and the other does not. No distributed algorithm can find a large cut in the first network. This will imply that the $T$-inefficiency of best responses is bounded away from the Price of Anarchy.

\begin{lemma} \label{lem:best-sg-alg}
	There is no randomized LOCAL algorithm that finds a cut with more than $(1/2 + 1/\sqrt{d-1})|E|$ edges in expectation in $T(n) = o(\log_d n)$ rounds on the bipartite networks from Lemma~\ref{lem:bipartite-high-girth}.
\end{lemma}

Lemma~\ref{lem:bipartite-high-girth} is presented in Appendix~\ref{app:graph-proofs}. The proof is similar to the proof of Lemma~\ref{lem:best-pgg-alg}. 
Since the networks $N$ from Lemma~\ref{lem:bipartite-high-girth} and $N'$ from Lemma~\ref{lem:small-cuts-high-girth} (also in Appendix~\ref{app:graph-proofs}) look locally the same to any distributed algorithm with running time $T = o(\log_d n)$, we can argue that expected size of the cut in on any $N$ is at most as large as the optimum solution on $N'$.

To estimate the worst-case behavior of best responses, we note that in the minority game best responses are \emph{monotone}, i.e.\ they never decrease the total utility.

\begin{lemma} \label{lem:ac-ne-half}
	The strategy profile computed by best responses from a random initial strategy profile has at least $|E|/2$ cut edges in expectation.
\end{lemma}

\begin{proof}
	First, note that a random strategy profile cuts exactly half of the edges in expectation: each edge has probability exactly 1/2 to be a cut edge.
	
	Now consider a best-response move by some agent $v$. Since $v$ is switching, it has at least one more cut edge in the new strategy profile. Since this change only affects edges around the agent, the total number of cut edges also increases by at least one.
\end{proof}

We are now ready to prove Theorem~\ref{thm:hd-poa}, estimating total welfare in the minority game.

\begin{proof}[Proof of Theorem~\ref{thm:hd-poa}]
	We again bound the three quantities.
	\begin{enumerate}
		\item Price of Anarchy. On the network $N$, the best solution cuts all edges. The total welfare is $2(d+1)n$. On the other hand, every Nash equilibrium cuts at least half of the edges: every agent has at least half of their neighbors on the other side of the cut. For even $d$ the total welfare in the equilibrium is at least $n$. The Price of Anarchy is $2(d+1)$ for even $d$.
		\item Best $T(n)$-time computable solution. By Lemma~\ref{lem:best-sg-alg}, the largest cut that can be computed in $o(\log_d n)$ rounds on network $N$ has at most $(1/2 + 1/\sqrt{d-1})|E|$ edges in it. The total welfare is at most $\bigl(1+2d/\sqrt{d-1}\bigr)n$ in expectation.
		\item Performance of best responses. By Lemma~\ref{lem:ac-ne-half}, best responses compute a cut of size at least $|E|/2$ in expectation. This gives an expected total welfare of at least $n$. 
	\end{enumerate}
	The $T(n)$-inefficiency of best responses is at most $1+2d/\sqrt{d-1}$ for $T(n) = o(\log_d n)$.
\end{proof}

\section{Mechanism design for best responses in graphical games} \label{sec:mechanism-design}

In this section we show that every graphical game $G$ with an efficiently solvable Nash equilibrium has a special property. It is possible to construct a related game $G'$ that we call a \emph{simulation game}: in $G'$ the best responses converge in one fair round and the equilibrium is equivalent to an equilibrium of $G$. The game is constructed so that the best responses simulate a distributed algorithm for computing a Nash equilibrium of the original game.

As the game $G'$ simulates a distributed algorithm, the simulation gains other properties of the algorithm as well. In particular, if there is an algorithm that computes a Nash equilibrium from some subset of equilibria with desirable properties, then best responses also converge to a Nash equilibrium from the same subset in the simulation game.

The simulation game can be constructed locally by a distributed algorithm. No similar game constructions exists for graphical games that \emph{do not have} efficiently solvable Nash equilibria.

\subsection{Constructing simulation games} \label{ssec:simulation-games}

To define simulation games, we need to consider algorithms in a specific \emph{normal form}, the existence of which is implied by the speedup result of Chang et al.\ \cite{Chang19exponential}. Lemma~\ref{lem:logstar-normal-form} states that $O(\log^* n)$-time algorithms can be decomposed into two phases. In the first phase the algorithm computes a distance-$(2t+2)$ coloring for some constant parameter $t$ that depends on the problem. Then a $t$-round algorithm is applied with the coloring as an input. We will construct games where the best responses construct these colorings and then choose the output of the algorithm on that particular coloring.

Now consider a game $G = (\A, u, N)$ that has an efficiently computable Nash equilibrium. Let $\F$ be a distributed algorithm in normal form that computes \emph{some} Nash equilibrium (that is, LVL $P$) of $G$ with the smallest possible constant running time $t$. Define a \emph{$t$-simulation game} $G' = (\A', u', N')$ of $G$ as follows.
\begin{enumerate}[noitemsep]
	\item The set of agents is the same as in $G$. In the network $N'$ connect two nodes $u$ and $v$ if and only if their distance in $N$ is at most $4t+2$.
	\item The actions $A'_v$ of each agent $v$ encode the possible locally correct simulations of $\F$. This is defined in two parts $A'_v = R_v \times \Sigma$. The first part $R_v$ consists of all possible labellings of the $t$-neighborhood of $v$ in $N$ with distinct colors from $\{1,2,\dots,\Delta^{2t+2}+1 \}$. The second part $\Sigma$ consists of the possible output labels of $P$. 
	Include the pair $(r, \sigma)$ in $A'_v$ if and only if $\F$ would output $\sigma$ on $v$ given $r$ as the input coloring of $B(v,t)$, where $B(v,t)$ denotes the $t$-hop neighborhood of $v$ in $N$. In addition there is the empty action.
	\item The utility functions $u_v$ encode the correct simulations. The utility $u_v(s) = 1$ if and only if the following hold. First, the coloring $r_v \in R_v$ is \emph{compatible} with the colorings $r_u \in R_u$ of each neighbor $u$ with a non-empty strategy. That is, for each $w \in B_N(v,t) \cap B_N(u,t)$, we have that $c_v(w) = c_u(w)$ or $c_u$ is empty. Second, the colorings form a proper $(2t+2)$-hop coloring of $N$. That is, if we map all the compatible colorings $c_v$ for all $v$ onto $N$, then two nodes in $N$ at distance at most $2t+2$ have distinct colors. Since two agents are connected in $N'$ if they are within distance $4t+2$ in $N$, it is possible to encode this in $u$. The color assigned to each agent at distance $t$ from some agent $v$ must differ from the colors of other agents within distance $2t+2$ of it. These agents can only be colored by agents within distance $4t+2$. 
	
	For the empty action and for strategy profiles that do not have these properties the utility is 0.
\end{enumerate}

We will show that simulation games converge in one fair round to an equilibrium that is equivalent to an equilibrium of the original game. We also want to show that similar constructions do not exist for games that do not have efficiently solvable Nash equilibria. Consider two graphical games $G = (\A, u, N)$ and $G' = (\A', u', N')$. We say that $G'$ is \emph{$k$-constructible given $G$} if there exists a $k$-round distributed algorithm $\F$ that on $N$, given $G$ as input computes $G'$ in the following sense. For each $v \in N'$, the algorithm can output the set of neighbors of $v$ in $N'$. In addition, for each $v$ it outputs the action set $A_v$ of $v$ and the utility function $u_v$ of $v$. 
We say that $G'$ \emph{corresponds} to $G$ if there exists a mapping $\varphi_v \colon A'_v \to A_v$ for all $v$ such that if $\vec{a}' = (a'_1, \dots, a'_n)$ is a Nash equilibrium of $G'$, then $\vec{a} = (\varphi(a'_1), \dots, \varphi(a'_n))$ is a Nash equilibrium of $G$.  

\begin{theorem} \label{thm:simulation-game}
	If a game $G$ has a Nash equilibrium that is solvable in time $O(\log^* n)$ as an LVL, then there is a $t$-simulation game $G' = (A', u', N')$ of $G = (A, u, N)$ with the following properties:
	\begin{enumerate}[noitemsep, label=(P\arabic*)]
		\item Best responses converge in one fair round from the empty initial strategy profile in $G'$. \label{p1:converge}
		\item $G'$ corresponds to $G$. \label{p2:correspondence}
		\item $G'$ is $(4t+2)$-constructible given $G$. \label{p3:constructability}
	\end{enumerate}
\end{theorem}

Properties~\ref{p1:converge} and \ref{p2:correspondence} imply that every game with an efficiently solvable Nash equilibrium has a $t$-simulation game that converges to an equivalent Nash equilibrium in one fair round. Property~\ref{p3:constructability} is in contrast with the following theorem which states that similar constructions do not exist not exist for games that do not have any efficiently solvable Nash equilibria.

\begin{theorem} \label{thm:no-sim-games}
	Let $G = (\A, u, N)$ be a game such that the Nash equilibria of $G$ as an LVL require $\Omega(T(n))$ time to compute in the deterministic LOCAL model, for any $T(n) = \Omega(\log^* n)$. Then there is no $k$-constructible corresponding game $G'$ given $G$, for any constant $k$, such that best responses converge in $o(T(n))$ rounds.
\end{theorem}

The proof of Theorem~\ref{thm:no-sim-games} follows from an application of our simulation theorem, Theorem~\ref{thm:local-dynamics-local}. If such games existed, a distributed algorithm could construct and simulate them, yielding a fast distributed algorithm for computing a Nash equilibrium of the original game.

The requirement of constructability is important, as without this property the game $G'$ could simply encode the structure of the Nash equilibria of $G$ in its actions or the utility function.

\subsection{Proving the properties of simulation games}

Before proving Theorem~\ref{thm:simulation-game}, we need the following technical lemma. It establishes that $O(\log^* n)$-time solvable LVLs can also be solved by algorithms in the required normal form. 

\begin{lemma} \label{lem:logstar-normal-form}
	Assume that LVL $P$ can be solved in $O(\log^* n)$ rounds by a deterministic distributed algorithm. Then there exists an algorithm $\F$ in the following normal form: the algorithm runs in $t = O(1)$ rounds (for some $t$ dependent on $P$), it takes a $(2t+2)$-hop $c$-coloring, for $c = \Delta^{2t+2}+1$, as an input, and outputs the solution to $P$.
\end{lemma} 

The proof follows from the speedup theorem of Chang et al.\ \cite{Chang19exponential}. A similar normal form construction has been used by Brandt et al.\ \cite{Brandt2017}.

We are now ready to prove the properties of simulation games.

\begin{proof}[Proof of Theorem~\ref{thm:simulation-game}]
	We show that the strategies of agents who have already played will constitute a correct partial simulation in $N'$. Initially this is trivially true, as all agents are assumed to start from the empty strategy.
	
	Now assume that $\vec{a}$ is the strategy profile after some number of best responses such that the non-empty strategies agree on the color of each agent and assume that some agent $v$ is scheduled to play. Since the utility is 0 if the agent chooses an action that is not compatible as a coloring, it must choose a compatible coloring. Since we assumed that $\vec{a}$ encodes a partial coloring and there are always enough colors to choose from (i.e. there are $\Delta^{2t+2}+1$ colors), $v$ can always choose an action that gives it utility 1. 
	
	After one fair round, each agent $v$ has chosen a strategy such that the colorings $r_v$ agree on all applicable nodes, and this is a Nash equilibrium of $G'$. The output label $\sigma_v$ of each action, by definition, corresponds to the output of $\F$ on $N$ given the input coloring that the agents have chosen. Now if we construct the correspondence $\varphi$ by mapping each action from $A'_v$ to $A_v$ by matching the output label, we have that if $\vec{a}$ is a Nash equilibrium of $G'$, then $\varphi(\vec{a}) = (\varphi(a_1), \dots, \varphi(a_n))$ is a Nash equilibrium of $G$, as it is the labelling computed by $\F$ on $N$. This is by definition a Nash equilibrium of $G$. This establishes properties \ref{p1:converge} and \ref{p2:correspondence}. 
	
	Finally, it remains to show that $G'$ can be constructed efficiently in the LOCAL model on the network $N$. This is achieved using a standard approach. First each node $v$ gathers its $(4t+2)$-hop neighborhood and outputs its neighbors in $N'$. Since the algorithm has access to $\F$, the algorithm in the normal form, it can consider every coloring of $B(v,t)$ and form the action set $A'_v$. Finally, since the algorithm has access to the $(4t+2)$-neighborhood of $v$ in $N$, it can compute the value of the utility function $u_v$ for all possible strategies of the neighbors.
\end{proof}

We now illustrate simulation games in the context of a specific graphical game.

\subsection{Case study: coloring game} \label{ssec:cs-coloring}

In this section we study a simple abstract coordination game called a \emph{coloring game}~\cite{Kearns06experimental}. The agents have to choose from a fixed set of resources and coordinate their actions so that they don't use the same resource as their neighbors.

Formally, in a \emph{$k$-coloring game} the actions of agent consist of $k$ colors $\{1,2,\dots,k\}$. The utility of each agent is 1 if no neighbor chooses the same color and 0 otherwise. The Nash equilibria of the coloring game correspond to (partial) colorings such that two neighbors choose the same color only if both have at least one neighbor of each other color.

We will show that it is possible that best responses fail to converge to a Nash equilibrium that corresponds to a coloring even if such a coloring \emph{exists and is efficiently computable}. This is naturally undesirable behavior, as the agents fail to solve the underlying coordination task. On the other hand, Theorem~\ref{thm:simulation-game} implies that in this case there exists a simulation game of the coloring game such that the best responses \emph{do converge} to Nash equilibrium that corresponds to a coloring.

We consider the coloring game on two-dimensional $n$-by-$n$ torus networks. That is, the set of nodes consists of $v_{i,j} \colon i,j \in \{0,1,\dots,n-1\}$, there is an edge between nodes $v_{i,j}$ and $v_{i,k}$ if $k = (j + 1) \bmod n$, and there is an edge between nodes $v_{i,j}$ and $v_{k,j}$ if $k = (i+1) \bmod n$. The complexity of $k$-coloring is completely understood in this setting.

\begin{theorem}[Brandt et al.\ \cite{Brandt2017}] \label{thm:grid-coloring}
	The complexities of $k$-coloring a two-dimensional torus are as follows.
	\begin{enumerate}[noitemsep]
		\item $k = 2,3 \colon \Theta(n)$, and
		\item $k = 4, 5, \dots \colon \Theta(\log^* n)$.
	\end{enumerate}
\end{theorem}

Both the 2-coloring problem and the 3-coloring problem are global (i.e.\ require $\Omega(n)$ rounds to compute). By Theorem~\ref{thm:local-dynamics-local} this implies that best responses also require at least $\Omega(n)$ rounds to converge.

In the case of $k \geq 5$, best responses always converge to a proper coloring. This is because each agent has degree $4$, so no matter what colors the neighbors choose, each agent can always choose a free color, which is their best response.

The interesting case is $k = 4$. Now the property that an agent always has a free color is no longer true: it can be that all four neighbors have the four different colors.

\begin{figure}[ht]
\centering
\includegraphics[width=0.35\textwidth]{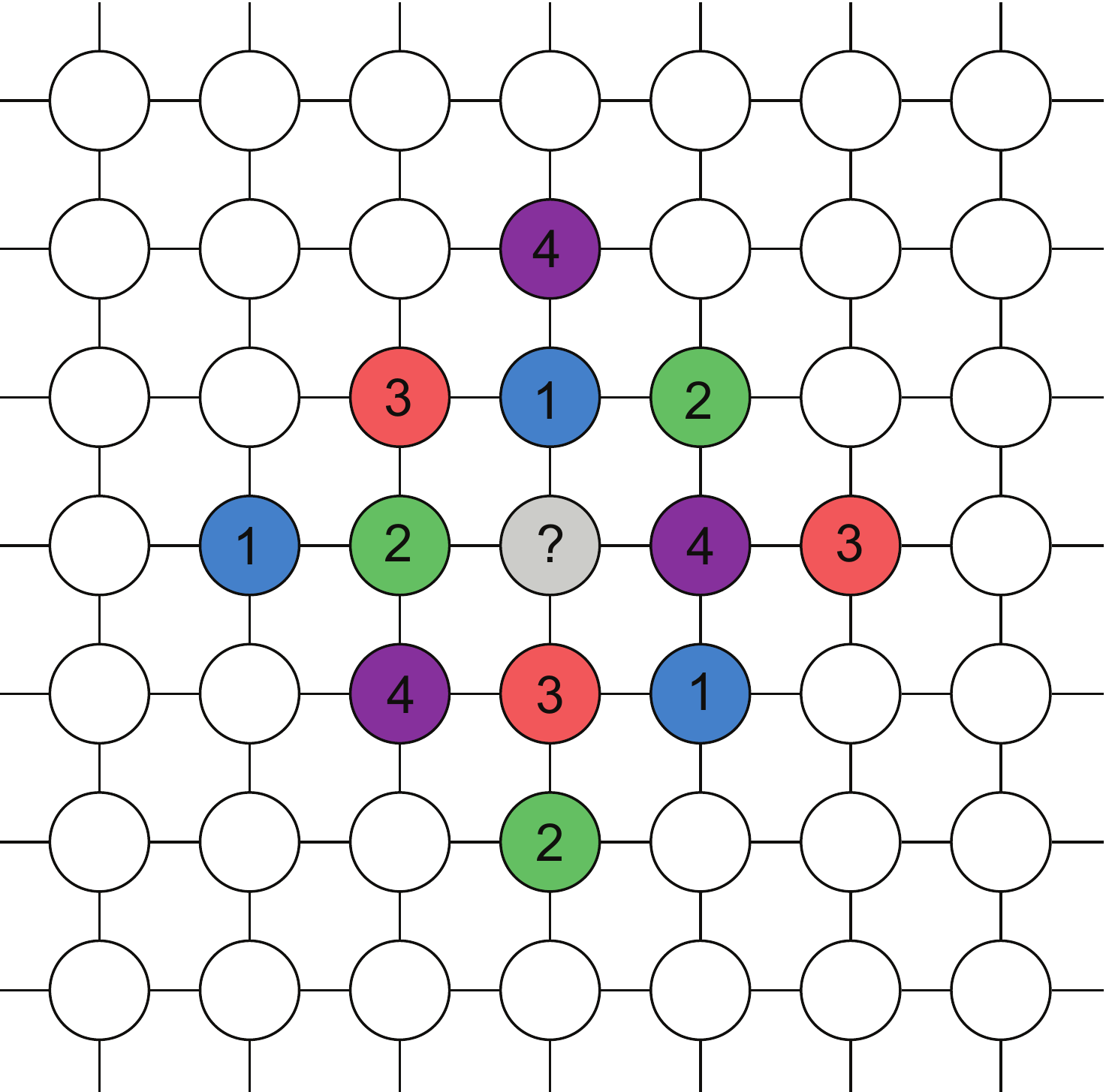}
\caption{\label{fig:fig2} Proof of Theorem~\ref{thm:coloring-bad-ne}. We show an initial coloring configuration on the grid that makes it impossible for best responses to converge to the proper 4-coloring. Every neighbor of the central node already has neighbors of all possible colors. This removes the incentive to switch.}
\end{figure}

\begin{theorem} \label{thm:coloring-bad-ne}
	There exists an initial strategy profile for every sufficiently large $n$-by-$n$ torus such that best responses do not converge to a 4-coloring in the 4-coloring game.
\end{theorem}

\begin{proof}
	Consider the initial configuration shown in Figure~\ref{fig:fig2}. The center node has no choice that would give a proper coloring. No matter which color it chooses, the neighbor with that color already has neighbors of the other colors. Therefore it has no incentive to switch, and the best-responses converge to a suboptimal equilibrium.
\end{proof}

In contrast, by Theorem~\ref{thm:grid-coloring} there exists an algorithm for 4-coloring the grid in $O(\log^* n)$ rounds. By Theorem~\ref{thm:simulation-game} this implies that there exists a simulation game where best responses converge to a 4-coloring.

\section{Conclusion}\label{sec:conclusion}

This paper introduced a novel approach to classifying the complexity of Nash equilibria in graphical games, and to understanding the convergence behavior of best-response and even general local dynamics. By establishing a connection to the analysis of distributed graph problems, we showed that the Nash equilibria of graphical games correspond to locally verifiable labelings: solutions to graph problems which are verifiable with constant round algorithms. Impossibility and complexity results provably transfer from the distributed setting to the game setting. Thus, we can leverage this to give lower bounds for convergence of best-response dynamics, to quantify the time-constrained inefficiency of best responses when convergence is slow or even absent, and to present how these results can be used for implementing mechanisms where best responses converge to a Nash equilibrium that is a solution of the corresponding graph problem. We exemplified our results with some simple and well-known graphical games. Our findings also relate to the open question of strategy proof algorithms for reaching equilibria in graphical games, as posed by Kearns in his 2007 survey. We note that in contrast to algorithms like in~\cite{Kearns2007survey}, our discussion indeed pertains to how agents \emph{reach} equilibria while playing the game in a way that is rational with respect to their locally restricted knowledge. This poses an interesting avenue for further research.

In this work, we have only considered pure Nash equilibria. However, we note that our approach is much more general than that, since it first and foremost depends on information being locally restricted without making additional assumptions like the game being a potential game. One can prove similar results as in this work for mixed Nash equilibria, but even for different equilibrium concepts altogether, such as correlated equilibria. The latter correspond to a model where information is local, but can in some sense be exchanged between neighboring nodes, introducing correlated strategy distributions. A further simple extension of our model could see different local strategy update dynamics (e.g. ficitious play) at work, instead of restricting analysis to best response dynamics only. This also highlights the connection of this research direction with evolutionary graph theory, which as a generalized approach to evolutionary dynamics features players with similarly bounded rationality~\cite{lieberman:Nature:2005}. Furthermore, it is also easily conceivable to analyze a far wider range of graphical games or LVLs with our approach, and to even extend the analysis to infinite graphs. 

Naturally, there are limitations to our approach. As we mention in Section 2, we assume that graphs are of bounded constant degree. Graphs of low diameter do not give impossibility results based on information propagation, as information can spread quickly. The results for such a setting would look quite different; furthermore, computation in a setting that comes closer to a centralized one where nodes do not have as strongly limited information is not well understood. We see this issue in the fact that the complexity gap for LVLs is a function of the maximum degree.
 Another issue lies with the fact that we only show lower bounds for convergence. This means that for instances that are not efficiently solvable, it could be the case that true convergence is far slower than what our results give.

However, it still holds that what we have considered in this paper should only represent a first taste for how powerful the connection between game theory and distributed computing is, and we note at this point that interpreting agents interacting in games on networks as a distributed system is also highly intuitive. There are many possibilities to harness this, and we have only explored a very limited number of them. We hope that this work can serve as a proof of concept, and be the starting point of exciting further research.

\begin{acks}
	This work was supported by the European Research Council (ERC) projects CoG 864228 (AdjustNet) and CoG 863818 (ForM-SMArt), and the Academy of Finland, grant 314888.
\end{acks}

\bibliographystyle{ACM-Reference-Format}
\bibliography{bibliography}


\begin{thebibliography}{48}


\ifx \showCODEN    \undefined \def \showCODEN     #1{\unskip}     \fi
\ifx \showDOI      \undefined \def \showDOI       #1{#1}\fi
\ifx \showISBNx    \undefined \def \showISBNx     #1{\unskip}     \fi
\ifx \showISBNxiii \undefined \def \showISBNxiii  #1{\unskip}     \fi
\ifx \showISSN     \undefined \def \showISSN      #1{\unskip}     \fi
\ifx \showLCCN     \undefined \def \showLCCN      #1{\unskip}     \fi
\ifx \shownote     \undefined \def \shownote      #1{#1}          \fi
\ifx \showarticletitle \undefined \def \showarticletitle #1{#1}   \fi
\ifx \showURL      \undefined \def \showURL       {\relax}        \fi
\providecommand\bibfield[2]{#2}
\providecommand\bibinfo[2]{#2}
\providecommand\natexlab[1]{#1}
\providecommand\showeprint[2][]{arXiv:#2}

\bibitem[\protect\citeauthoryear{Abramson and Kuperman}{Abramson and
  Kuperman}{2001}]%
        {abramson2001social}
\bibfield{author}{\bibinfo{person}{Guillermo Abramson} {and}
  \bibinfo{person}{Marcelo Kuperman}.} \bibinfo{year}{2001}\natexlab{}.
\newblock \showarticletitle{Social games in a social network}.
\newblock \bibinfo{journal}{\emph{Physical Review E}} \bibinfo{volume}{63},
  \bibinfo{number}{3} (\bibinfo{year}{2001}), \bibinfo{pages}{030901}.
\newblock


\bibitem[\protect\citeauthoryear{Alon}{Alon}{2010}]%
        {alon10constant}
\bibfield{author}{\bibinfo{person}{Noga Alon}.}
  \bibinfo{year}{2010}\natexlab{}.
\newblock \bibinfo{booktitle}{\emph{On Constant Time Approximation of
  Parameters of Bounded Degree Graphs}}.
\newblock \bibinfo{publisher}{Springer-Verlag}, \bibinfo{address}{Berlin,
  Heidelberg}, \bibinfo{pages}{234–239}.
\newblock


\bibitem[\protect\citeauthoryear{Alon and Wormald}{Alon and Wormald}{2010}]%
        {Alon10star}
\bibfield{author}{\bibinfo{person}{Noga Alon} {and} \bibinfo{person}{Nicholas
  Wormald}.} \bibinfo{year}{2010}\natexlab{}.
\newblock \bibinfo{booktitle}{\emph{High Degree Graphs Contain Large-Star
  Factors}}.
\newblock \bibinfo{publisher}{Springer Berlin Heidelberg},
  \bibinfo{address}{Berlin, Heidelberg}, \bibinfo{pages}{9--21}.
\newblock
\showISBNx{978-3-642-13580-4}
\urldef\tempurl%
\url{https://doi.org/10.1007/978-3-642-13580-4_1}
\showDOI{\tempurl}


\bibitem[\protect\citeauthoryear{Aspnes, Chang, and Yampolskiy}{Aspnes
  et~al\mbox{.}}{2006}]%
        {Aspnes2006}
\bibfield{author}{\bibinfo{person}{James Aspnes}, \bibinfo{person}{Kevin
  Chang}, {and} \bibinfo{person}{Aleksandr Yampolskiy}.}
  \bibinfo{year}{2006}\natexlab{}.
\newblock \showarticletitle{Inoculation strategies for victims of viruses and
  the sum-of-squares partition problem}.
\newblock \bibinfo{journal}{\emph{J. Comput. System Sci.}}
  \bibinfo{volume}{72}, \bibinfo{number}{6} (\bibinfo{year}{2006}),
  \bibinfo{pages}{1077--1093}.
\newblock


\bibitem[\protect\citeauthoryear{Avni, Henzinger, and Kupferman}{Avni
  et~al\mbox{.}}{2016}]%
        {Avni2016}
\bibfield{author}{\bibinfo{person}{Guy Avni}, \bibinfo{person}{Thomas~A.
  Henzinger}, {and} \bibinfo{person}{Orna Kupferman}.}
  \bibinfo{year}{2016}\natexlab{}.
\newblock \showarticletitle{Dynamic Resource Allocation Games}. In
  \bibinfo{booktitle}{\emph{Algorithmic Game Theory}}.
  \bibinfo{publisher}{Springer Berlin Heidelberg}, \bibinfo{pages}{153--166}.
\newblock


\bibitem[\protect\citeauthoryear{Balliu, Brandt, Efron, Hirvonen, Maus,
  Olivetti, and Suomela}{Balliu et~al\mbox{.}}{2019a}]%
        {balliu19binary}
\bibfield{author}{\bibinfo{person}{Alkida Balliu}, \bibinfo{person}{Sebastian
  Brandt}, \bibinfo{person}{Yuval Efron}, \bibinfo{person}{Juho Hirvonen},
  \bibinfo{person}{Yannic Maus}, \bibinfo{person}{Dennis Olivetti}, {and}
  \bibinfo{person}{Jukka Suomela}.} \bibinfo{year}{2019}\natexlab{a}.
\newblock \showarticletitle{Classification of distributed binary labeling
  problems}.
\newblock \bibinfo{journal}{\emph{CoRR}}  \bibinfo{volume}{abs/1911.13294}
  (\bibinfo{year}{2019}).
\newblock
\showeprint[arxiv]{1911.13294}
\urldef\tempurl%
\url{http://arxiv.org/abs/1911.13294}
\showURL{%
\tempurl}


\bibitem[\protect\citeauthoryear{Balliu, Brandt, Hirvonen, Olivetti, Rabie, and
  Suomela}{Balliu et~al\mbox{.}}{2019b}]%
        {balliu19maximal}
\bibfield{author}{\bibinfo{person}{Alkida Balliu}, \bibinfo{person}{Sebastian
  Brandt}, \bibinfo{person}{Juho Hirvonen}, \bibinfo{person}{Dennis Olivetti},
  \bibinfo{person}{Mika\"{e}l Rabie}, {and} \bibinfo{person}{Jukka Suomela}.}
  \bibinfo{year}{2019}\natexlab{b}.
\newblock \showarticletitle{Lower bounds for maximal matchings and maximal
  independent sets}. In \bibinfo{booktitle}{\emph{Proc. 60th IEEE Symposium on
  Foundations of Computer Science (FOCS 2019)}}.
\newblock
\newblock
\shownote{\href{https://arxiv.org/abs/1901.02441}{\footnotesize\sf
  arXiv:1901.02441}.}


\bibitem[\protect\citeauthoryear{Balliu, Hirvonen, Korhonen, Lempi{\"{a}}inen,
  Olivetti, and Suomela}{Balliu et~al\mbox{.}}{2018}]%
        {Balliu2018stoc}
\bibfield{author}{\bibinfo{person}{Alkida Balliu}, \bibinfo{person}{Juho
  Hirvonen}, \bibinfo{person}{Janne~H Korhonen}, \bibinfo{person}{Tuomo
  Lempi{\"{a}}inen}, \bibinfo{person}{Dennis Olivetti}, {and}
  \bibinfo{person}{Jukka Suomela}.} \bibinfo{year}{2018}\natexlab{}.
\newblock \showarticletitle{{New classes of distributed time complexity}}. In
  \bibinfo{booktitle}{\emph{Proc. 50th ACM Symposium on Theory of Computing
  (STOC 2018)}}. \bibinfo{publisher}{ACM Press}, \bibinfo{pages}{1307--1318}.
\newblock
\urldef\tempurl%
\url{https://doi.org/10.1145/3188745.3188860}
\showDOI{\tempurl}


\bibitem[\protect\citeauthoryear{Balliu, Hirvonen, Lenzen, Olivetti, and
  Suomela}{Balliu et~al\mbox{.}}{2019c}]%
        {balliu19weak}
\bibfield{author}{\bibinfo{person}{Alkida Balliu}, \bibinfo{person}{Juho
  Hirvonen}, \bibinfo{person}{Christoph Lenzen}, \bibinfo{person}{Dennis
  Olivetti}, {and} \bibinfo{person}{Jukka Suomela}.}
  \bibinfo{year}{2019}\natexlab{c}.
\newblock \showarticletitle{Locality of Not-so-Weak Coloring}. In
  \bibinfo{booktitle}{\emph{Proc. 26th International Colloquium on Structural
  Information and Communication Complexity ({SIROCCO} 2019)}}
  \emph{(\bibinfo{series}{Lecture Notes in Computer Science},
  Vol.~\bibinfo{volume}{11639})}. \bibinfo{publisher}{Springer},
  \bibinfo{pages}{37--51}.
\newblock
\urldef\tempurl%
\url{https://doi.org/10.1007/978-3-030-24922-9\_3}
\showDOI{\tempurl}


\bibitem[\protect\citeauthoryear{Barenboim}{Barenboim}{2016}]%
        {Barenboim16deterministic}
\bibfield{author}{\bibinfo{person}{Leonid Barenboim}.}
  \bibinfo{year}{2016}\natexlab{}.
\newblock \showarticletitle{Deterministic ({\(\Delta\)} + 1)-Coloring in
  Sublinear (in {\(\Delta\)}) Time in Static, Dynamic, and Faulty Networks}.
\newblock \bibinfo{journal}{\emph{J. {ACM}}} \bibinfo{volume}{63},
  \bibinfo{number}{5} (\bibinfo{year}{2016}), \bibinfo{pages}{47:1--47:22}.
\newblock
\urldef\tempurl%
\url{https://doi.org/10.1145/2979675}
\showDOI{\tempurl}


\bibitem[\protect\citeauthoryear{Bhawalkar, Gairing, and Roughgarden}{Bhawalkar
  et~al\mbox{.}}{2014}]%
        {bhawalkar2014weighted}
\bibfield{author}{\bibinfo{person}{Kshipra Bhawalkar}, \bibinfo{person}{Martin
  Gairing}, {and} \bibinfo{person}{Tim Roughgarden}.}
  \bibinfo{year}{2014}\natexlab{}.
\newblock \showarticletitle{Weighted congestion games: the price of anarchy,
  universal worst-case examples, and tightness}.
\newblock \bibinfo{journal}{\emph{ACM Transactions on Economics and Computation
  (TEAC)}} \bibinfo{volume}{2}, \bibinfo{number}{4} (\bibinfo{year}{2014}),
  \bibinfo{pages}{1--23}.
\newblock


\bibitem[\protect\citeauthoryear{Bollobas}{Bollobas}{2004}]%
        {bollobas04extremal}
\bibfield{author}{\bibinfo{person}{B\'{e}la Bollobas}.}
  \bibinfo{year}{2004}\natexlab{}.
\newblock \bibinfo{booktitle}{\emph{Extremal Graph Theory}}.
\newblock \bibinfo{publisher}{Dover Publications, Inc.},
  \bibinfo{address}{USA}.
\newblock


\bibitem[\protect\citeauthoryear{Bramoull{\'e}}{Bramoull{\'e}}{2007}]%
        {Bramouille07social}
\bibfield{author}{\bibinfo{person}{Yann Bramoull{\'e}}.}
  \bibinfo{year}{2007}\natexlab{}.
\newblock \showarticletitle{Anti-coordination and social interactions}.
\newblock \bibinfo{journal}{\emph{Games and Economic Behavior}}
  \bibinfo{volume}{58}, \bibinfo{number}{1} (\bibinfo{year}{2007}),
  \bibinfo{pages}{30--49}.
\newblock
\showISSN{0899-8256}
\urldef\tempurl%
\url{https://doi.org/10.1016/j.geb.2005.12.006}
\showDOI{\tempurl}


\bibitem[\protect\citeauthoryear{Bramoull{\'e} and Kranton}{Bramoull{\'e} and
  Kranton}{2007}]%
        {bramoulle07public}
\bibfield{author}{\bibinfo{person}{Yann Bramoull{\'e}} {and}
  \bibinfo{person}{Rachel Kranton}.} \bibinfo{year}{2007}\natexlab{}.
\newblock \showarticletitle{Public goods in networks}.
\newblock \bibinfo{journal}{\emph{Journal of Economic Theory}}
  \bibinfo{volume}{135}, \bibinfo{number}{1} (\bibinfo{year}{2007}),
  \bibinfo{pages}{478 -- 494}.
\newblock
\showISSN{0022-0531}
\urldef\tempurl%
\url{https://doi.org/10.1016/j.jet.2006.06.006}
\showDOI{\tempurl}


\bibitem[\protect\citeauthoryear{Brandt}{Brandt}{2019}]%
        {brandt19automatic}
\bibfield{author}{\bibinfo{person}{Sebastian Brandt}.}
  \bibinfo{year}{2019}\natexlab{}.
\newblock \bibinfo{title}{An Automatic Speedup Theorem for Distributed
  Problems}.
\newblock
\newblock
\showeprint[arxiv]{1902.09958}
\urldef\tempurl%
\url{https://arxiv.org/abs/1902.09958v1}
\showURL{%
\tempurl}


\bibitem[\protect\citeauthoryear{Brandt, Fischer, Hirvonen, Keller,
  Lempi{\"{a}}inen, Rybicki, Suomela, and Uitto}{Brandt et~al\mbox{.}}{2016}]%
        {Brandt2016}
\bibfield{author}{\bibinfo{person}{Sebastian Brandt}, \bibinfo{person}{Orr
  Fischer}, \bibinfo{person}{Juho Hirvonen}, \bibinfo{person}{Barbara Keller},
  \bibinfo{person}{Tuomo Lempi{\"{a}}inen}, \bibinfo{person}{Joel Rybicki},
  \bibinfo{person}{Jukka Suomela}, {and} \bibinfo{person}{Jara Uitto}.}
  \bibinfo{year}{2016}\natexlab{}.
\newblock \showarticletitle{{A lower bound for the distributed Lov{\'{a}}sz
  local lemma}}. In \bibinfo{booktitle}{\emph{Proc. 48th ACM Symposium on
  Theory of Computing (STOC 2016)}}. \bibinfo{publisher}{ACM Press},
  \bibinfo{pages}{479--488}.
\newblock
\urldef\tempurl%
\url{https://doi.org/10.1145/2897518.2897570}
\showDOI{\tempurl}


\bibitem[\protect\citeauthoryear{Brandt, Hirvonen, Korhonen, Lempi{\"{a}}inen,
  {\"{O}}sterg{\aa}rd, Purcell, Rybicki, Suomela, and Uzna{\'{n}}ski}{Brandt
  et~al\mbox{.}}{2017}]%
        {Brandt2017}
\bibfield{author}{\bibinfo{person}{Sebastian Brandt}, \bibinfo{person}{Juho
  Hirvonen}, \bibinfo{person}{Janne~H Korhonen}, \bibinfo{person}{Tuomo
  Lempi{\"{a}}inen}, \bibinfo{person}{Patric R~J {\"{O}}sterg{\aa}rd},
  \bibinfo{person}{Christopher Purcell}, \bibinfo{person}{Joel Rybicki},
  \bibinfo{person}{Jukka Suomela}, {and} \bibinfo{person}{Przemys{\l}aw
  Uzna{\'{n}}ski}.} \bibinfo{year}{2017}\natexlab{}.
\newblock \showarticletitle{{LCL problems on grids}}. In
  \bibinfo{booktitle}{\emph{Proc. 36th ACM Symposium on Principles of
  Distributed Computing (PODC 2017)}}. \bibinfo{publisher}{ACM Press},
  \bibinfo{pages}{101--110}.
\newblock
\urldef\tempurl%
\url{https://doi.org/10.1145/3087801.3087833}
\showDOI{\tempurl}


\bibitem[\protect\citeauthoryear{Challet and Zhang}{Challet and Zhang}{1997}]%
        {challet1997emergence}
\bibfield{author}{\bibinfo{person}{Damien Challet} {and} \bibinfo{person}{Y-C
  Zhang}.} \bibinfo{year}{1997}\natexlab{}.
\newblock \showarticletitle{Emergence of cooperation and organization in an
  evolutionary game}.
\newblock \bibinfo{journal}{\emph{Physica A: Statistical Mechanics and its
  Applications}} \bibinfo{volume}{246}, \bibinfo{number}{3-4}
  (\bibinfo{year}{1997}), \bibinfo{pages}{407--418}.
\newblock


\bibitem[\protect\citeauthoryear{Chang, Kopelowitz, and Pettie}{Chang
  et~al\mbox{.}}{2019}]%
        {Chang19exponential}
\bibfield{author}{\bibinfo{person}{Yi{-}Jun Chang}, \bibinfo{person}{Tsvi
  Kopelowitz}, {and} \bibinfo{person}{Seth Pettie}.}
  \bibinfo{year}{2019}\natexlab{}.
\newblock \showarticletitle{An Exponential Separation between Randomized and
  Deterministic Complexity in the {LOCAL} Model}.
\newblock \bibinfo{journal}{\emph{{SIAM} J. Comput.}} \bibinfo{volume}{48},
  \bibinfo{number}{1} (\bibinfo{year}{2019}), \bibinfo{pages}{122--143}.
\newblock
\urldef\tempurl%
\url{https://doi.org/10.1137/17M1117537}
\showDOI{\tempurl}


\bibitem[\protect\citeauthoryear{Chang and Pettie}{Chang and Pettie}{2019}]%
        {Chang19time}
\bibfield{author}{\bibinfo{person}{Yi{-}Jun Chang} {and} \bibinfo{person}{Seth
  Pettie}.} \bibinfo{year}{2019}\natexlab{}.
\newblock \showarticletitle{A Time Hierarchy Theorem for the {LOCAL} Model}.
\newblock \bibinfo{journal}{\emph{{SIAM} J. Comput.}} \bibinfo{volume}{48},
  \bibinfo{number}{1} (\bibinfo{year}{2019}), \bibinfo{pages}{33--69}.
\newblock
\urldef\tempurl%
\url{https://doi.org/10.1137/17M1157957}
\showDOI{\tempurl}


\bibitem[\protect\citeauthoryear{Chang, He, Li, Pettie, and Uitto}{Chang
  et~al\mbox{.}}{2018}]%
        {chang18complexity}
\bibfield{author}{\bibinfo{person}{Yi-Jun Chang}, \bibinfo{person}{Qizheng He},
  \bibinfo{person}{Wenzheng Li}, \bibinfo{person}{Seth Pettie}, {and}
  \bibinfo{person}{Jara Uitto}.} \bibinfo{year}{2018}\natexlab{}.
\newblock \showarticletitle{{The Complexity of Distributed Edge Coloring with
  Small Palettes}}. In \bibinfo{booktitle}{\emph{Proc. 29th ACM-SIAM Symposium
  on Discrete Algorithms (SODA 2018)}}. \bibinfo{publisher}{Society for
  Industrial and Applied Mathematics}, \bibinfo{pages}{2633--2652}.
\newblock
\urldef\tempurl%
\url{https://doi.org/10.1137/1.9781611975031.168}
\showDOI{\tempurl}


\bibitem[\protect\citeauthoryear{Cherry, Cotten, and Kroll}{Cherry
  et~al\mbox{.}}{2013}]%
        {cherry2013heterogeneity}
\bibfield{author}{\bibinfo{person}{Todd~L Cherry}, \bibinfo{person}{Stephen~J
  Cotten}, {and} \bibinfo{person}{Stephan Kroll}.}
  \bibinfo{year}{2013}\natexlab{}.
\newblock \showarticletitle{Heterogeneity, coordination and the provision of
  best-shot public goods}.
\newblock \bibinfo{journal}{\emph{Experimental Economics}}
  \bibinfo{volume}{16}, \bibinfo{number}{4} (\bibinfo{year}{2013}),
  \bibinfo{pages}{497--510}.
\newblock


\bibitem[\protect\citeauthoryear{Collet, Fraigniaud, and Penna}{Collet
  et~al\mbox{.}}{2018}]%
        {collet2018equilibria}
\bibfield{author}{\bibinfo{person}{Simon Collet}, \bibinfo{person}{Pierre
  Fraigniaud}, {and} \bibinfo{person}{Paolo Penna}.}
  \bibinfo{year}{2018}\natexlab{}.
\newblock \showarticletitle{Equilibria of games in networks for local tasks}.
  In \bibinfo{booktitle}{\emph{22nd International Conference on Principles of
  Distributed Systems (OPODIS 2018)}}. Schloss Dagstuhl-Leibniz-Zentrum fuer
  Informatik.
\newblock


\bibitem[\protect\citeauthoryear{Daskalakis and Papadimitriou}{Daskalakis and
  Papadimitriou}{2006}]%
        {Daskalakis2006}
\bibfield{author}{\bibinfo{person}{Constantinos Daskalakis} {and}
  \bibinfo{person}{Christos~H Papadimitriou}.} \bibinfo{year}{2006}\natexlab{}.
\newblock \showarticletitle{Computing pure Nash equilibria in graphical games
  via Markov random fields}. In \bibinfo{booktitle}{\emph{Proceedings of the
  7th ACM Conference on Electronic Commerce}}. \bibinfo{pages}{91--99}.
\newblock


\bibitem[\protect\citeauthoryear{Elkind, Goldberg, and Goldberg}{Elkind
  et~al\mbox{.}}{2006}]%
        {Elkind2006}
\bibfield{author}{\bibinfo{person}{Edith Elkind}, \bibinfo{person}{Leslie~Ann
  Goldberg}, {and} \bibinfo{person}{Paul Goldberg}.}
  \bibinfo{year}{2006}\natexlab{}.
\newblock \showarticletitle{Nash Equilibria in Graphical Games on Trees
  Revisited}. In \bibinfo{booktitle}{\emph{Proceedings of the 7th ACM
  Conference on Electronic Commerce}} \emph{(\bibinfo{series}{EC '06})}.
  \bibinfo{publisher}{Association for Computing Machinery},
  \bibinfo{pages}{100–109}.
\newblock
\urldef\tempurl%
\url{https://doi.org/10.1145/1134707.1134719}
\showDOI{\tempurl}


\bibitem[\protect\citeauthoryear{Friedman}{Friedman}{2003}]%
        {Friedman03}
\bibfield{author}{\bibinfo{person}{Joel Friedman}.}
  \bibinfo{year}{2003}\natexlab{}.
\newblock \showarticletitle{A proof of Alon's second eigenvalue conjecture}. In
  \bibinfo{booktitle}{\emph{Proc. 35th Annual {ACM} Symposium on Theory of
  Computing (STOC 2003)}}. \bibinfo{publisher}{{ACM}},
  \bibinfo{pages}{720--724}.
\newblock
\urldef\tempurl%
\url{https://doi.org/10.1145/780542.780646}
\showDOI{\tempurl}


\bibitem[\protect\citeauthoryear{Frieze and Łuczak}{Frieze and
  Łuczak}{1992}]%
        {FRIEZE1992}
\bibfield{author}{\bibinfo{person}{A.M Frieze} {and} \bibinfo{person}{T
  Łuczak}.} \bibinfo{year}{1992}\natexlab{}.
\newblock \showarticletitle{On the independence and chromatic numbers of random
  regular graphs}.
\newblock \bibinfo{journal}{\emph{Journal of Combinatorial Theory, Series B}}
  \bibinfo{volume}{54}, \bibinfo{number}{1} (\bibinfo{year}{1992}),
  \bibinfo{pages}{123--132}.
\newblock
\showISSN{0095-8956}
\urldef\tempurl%
\url{https://doi.org/10.1016/0095-8956(92)90070-E}
\showDOI{\tempurl}


\bibitem[\protect\citeauthoryear{Ghaffari and Su}{Ghaffari and Su}{2017}]%
        {ghaffari17distributed}
\bibfield{author}{\bibinfo{person}{Mohsen Ghaffari} {and}
  \bibinfo{person}{Hsin-Hao Su}.} \bibinfo{year}{2017}\natexlab{}.
\newblock \showarticletitle{{Distributed Degree Splitting, Edge Coloring, and
  Orientations}}. In \bibinfo{booktitle}{\emph{Proc. 28th ACM-SIAM Symposium on
  Discrete Algorithms (SODA 2017)}}. \bibinfo{publisher}{Society for Industrial
  and Applied Mathematics}, \bibinfo{pages}{2505--2523}.
\newblock
\urldef\tempurl%
\url{https://doi.org/10.1137/1.9781611974782.166}
\showDOI{\tempurl}


\bibitem[\protect\citeauthoryear{Jackson and Yariv}{Jackson and Yariv}{2007}]%
        {jackson2007diffusion}
\bibfield{author}{\bibinfo{person}{Matthew~O Jackson} {and}
  \bibinfo{person}{Leeat Yariv}.} \bibinfo{year}{2007}\natexlab{}.
\newblock \showarticletitle{Diffusion of behavior and equilibrium properties in
  network games}.
\newblock \bibinfo{journal}{\emph{American Economic Review}}
  \bibinfo{volume}{97}, \bibinfo{number}{2} (\bibinfo{year}{2007}),
  \bibinfo{pages}{92--98}.
\newblock


\bibitem[\protect\citeauthoryear{Jackson and Zenou}{Jackson and Zenou}{2015}]%
        {jackson2015games}
\bibfield{author}{\bibinfo{person}{Matthew~O Jackson} {and}
  \bibinfo{person}{Yves Zenou}.} \bibinfo{year}{2015}\natexlab{}.
\newblock \showarticletitle{Games on networks}.
\newblock In \bibinfo{booktitle}{\emph{Handbook of game theory with economic
  applications}}. Vol.~\bibinfo{volume}{4}. \bibinfo{publisher}{Elsevier},
  \bibinfo{pages}{95--163}.
\newblock


\bibitem[\protect\citeauthoryear{Kearns}{Kearns}{2007}]%
        {Kearns2007survey}
\bibfield{author}{\bibinfo{person}{Michael Kearns}.}
  \bibinfo{year}{2007}\natexlab{}.
\newblock \showarticletitle{Graphical games}.
\newblock \bibinfo{journal}{\emph{Algorithmic game theory}}
  \bibinfo{volume}{3} (\bibinfo{year}{2007}), \bibinfo{pages}{159--180}.
\newblock


\bibitem[\protect\citeauthoryear{Kearns, Suri, and Montfort}{Kearns
  et~al\mbox{.}}{2006}]%
        {Kearns06experimental}
\bibfield{author}{\bibinfo{person}{Michael Kearns}, \bibinfo{person}{Siddharth
  Suri}, {and} \bibinfo{person}{Nick Montfort}.}
  \bibinfo{year}{2006}\natexlab{}.
\newblock \showarticletitle{An Experimental Study of the Coloring Problem on
  Human Subject Networks}.
\newblock \bibinfo{journal}{\emph{Science}} \bibinfo{volume}{313},
  \bibinfo{number}{5788} (\bibinfo{year}{2006}), \bibinfo{pages}{824--827}.
\newblock
\urldef\tempurl%
\url{https://doi.org/10.1126/science.1127207}
\showDOI{\tempurl}


\bibitem[\protect\citeauthoryear{Kearns, Littman, and Singh}{Kearns
  et~al\mbox{.}}{2001}]%
        {kearns01graphical}
\bibfield{author}{\bibinfo{person}{Michael~J. Kearns},
  \bibinfo{person}{Michael~L. Littman}, {and} \bibinfo{person}{Satinder~P.
  Singh}.} \bibinfo{year}{2001}\natexlab{}.
\newblock \showarticletitle{Graphical Models for Game Theory}. In
  \bibinfo{booktitle}{\emph{Proc. 17th Conference in Uncertainty in Artificial
  Intelligence (UAI 2001)}}. \bibinfo{publisher}{Morgan Kaufmann Publishers
  Inc.}, \bibinfo{address}{San Francisco, CA, USA}, \bibinfo{pages}{253–260}.
\newblock


\bibitem[\protect\citeauthoryear{Komarovsky, Levit, Grinshpoun, and
  Meisels}{Komarovsky et~al\mbox{.}}{2015}]%
        {Komarovsky15}
\bibfield{author}{\bibinfo{person}{Zohar Komarovsky}, \bibinfo{person}{Vadim
  Levit}, \bibinfo{person}{Tal Grinshpoun}, {and} \bibinfo{person}{Amnon
  Meisels}.} \bibinfo{year}{2015}\natexlab{}.
\newblock \showarticletitle{Efficient Equilibria in a Public Goods Game}. In
  \bibinfo{booktitle}{\emph{Proceedings of the 2015 IEEE / WIC / ACM
  International Conference on Web Intelligence and Intelligent Agent Technology
  (WI-IAT) - Volume 01}} \emph{(\bibinfo{series}{WI-IAT '15})}.
  \bibinfo{publisher}{IEEE Computer Society}, \bibinfo{pages}{214–219}.
\newblock
\showISBNx{9781467396189}
\urldef\tempurl%
\url{https://doi.org/10.1109/WI-IAT.2015.91}
\showDOI{\tempurl}


\bibitem[\protect\citeauthoryear{Korman, Kutten, and Peleg}{Korman
  et~al\mbox{.}}{2010}]%
        {korman2010proof}
\bibfield{author}{\bibinfo{person}{Amos Korman}, \bibinfo{person}{Shay Kutten},
  {and} \bibinfo{person}{David Peleg}.} \bibinfo{year}{2010}\natexlab{}.
\newblock \showarticletitle{Proof labeling schemes}.
\newblock \bibinfo{journal}{\emph{Distributed Computing}} \bibinfo{volume}{22},
  \bibinfo{number}{4} (\bibinfo{year}{2010}), \bibinfo{pages}{215--233}.
\newblock


\bibitem[\protect\citeauthoryear{Koutsoupias and Papadimitriou}{Koutsoupias and
  Papadimitriou}{1999}]%
        {KoutsoupiasP99}
\bibfield{author}{\bibinfo{person}{Elias Koutsoupias} {and}
  \bibinfo{person}{Christos~H. Papadimitriou}.}
  \bibinfo{year}{1999}\natexlab{}.
\newblock \showarticletitle{Worst-case Equilibria}. In
  \bibinfo{booktitle}{\emph{Proc. 16th Annual Symposium on Theoretical Aspects
  of Computer Science ({STACS} 1999)}} \emph{(\bibinfo{series}{Lecture Notes in
  Computer Science}, Vol.~\bibinfo{volume}{1563})}.
  \bibinfo{publisher}{Springer}, \bibinfo{pages}{404--413}.
\newblock
\urldef\tempurl%
\url{https://doi.org/10.1007/3-540-49116-3\_38}
\showDOI{\tempurl}


\bibitem[\protect\citeauthoryear{Lieberman, Hauert, and Nowak}{Lieberman
  et~al\mbox{.}}{2005}]%
        {lieberman:Nature:2005}
\bibfield{author}{\bibinfo{person}{E. Lieberman}, \bibinfo{person}{C. Hauert},
  {and} \bibinfo{person}{M.~A. Nowak}.} \bibinfo{year}{2005}\natexlab{}.
\newblock \showarticletitle{Evolutionary dynamics on graphs.}
\newblock \bibinfo{journal}{\emph{Nature}}  \bibinfo{volume}{433}
  (\bibinfo{year}{2005}), \bibinfo{pages}{312--316}.
\newblock


\bibitem[\protect\citeauthoryear{Linial}{Linial}{1992}]%
        {Linial1992}
\bibfield{author}{\bibinfo{person}{Nathan Linial}.}
  \bibinfo{year}{1992}\natexlab{}.
\newblock \showarticletitle{{Locality in Distributed Graph Algorithms}}.
\newblock \bibinfo{journal}{\emph{SIAM J. Comput.}} \bibinfo{volume}{21},
  \bibinfo{number}{1} (\bibinfo{year}{1992}), \bibinfo{pages}{193--201}.
\newblock
\urldef\tempurl%
\url{https://doi.org/10.1137/0221015}
\showDOI{\tempurl}


\bibitem[\protect\citeauthoryear{Naor and Stockmeyer}{Naor and
  Stockmeyer}{1995}]%
        {Naor1995}
\bibfield{author}{\bibinfo{person}{Moni Naor} {and} \bibinfo{person}{Larry
  Stockmeyer}.} \bibinfo{year}{1995}\natexlab{}.
\newblock \showarticletitle{{What Can be Computed Locally?}}
\newblock \bibinfo{journal}{\emph{SIAM J. Comput.}} \bibinfo{volume}{24},
  \bibinfo{number}{6} (\bibinfo{year}{1995}), \bibinfo{pages}{1259--1277}.
\newblock
\urldef\tempurl%
\url{https://doi.org/10.1137/S0097539793254571}
\showDOI{\tempurl}


\bibitem[\protect\citeauthoryear{Ortiz and Kearns}{Ortiz and Kearns}{2003}]%
        {Ortiz2003}
\bibfield{author}{\bibinfo{person}{Luis~E Ortiz} {and} \bibinfo{person}{Michael
  Kearns}.} \bibinfo{year}{2003}\natexlab{}.
\newblock \showarticletitle{Nash propagation for loopy graphical games}.
\newblock \bibinfo{journal}{\emph{Advances in neural information processing
  systems}} (\bibinfo{year}{2003}), \bibinfo{pages}{817--824}.
\newblock


\bibitem[\protect\citeauthoryear{Panconesi and Rizzi}{Panconesi and
  Rizzi}{2001}]%
        {panconesi01simple}
\bibfield{author}{\bibinfo{person}{Alessandro Panconesi} {and}
  \bibinfo{person}{Romeo Rizzi}.} \bibinfo{year}{2001}\natexlab{}.
\newblock \showarticletitle{{Some simple distributed algorithms for sparse
  networks}}.
\newblock \bibinfo{journal}{\emph{Distributed Computing}} \bibinfo{volume}{14},
  \bibinfo{number}{2} (\bibinfo{year}{2001}), \bibinfo{pages}{97--100}.
\newblock
\urldef\tempurl%
\url{https://doi.org/10.1007/PL00008932}
\showDOI{\tempurl}


\bibitem[\protect\citeauthoryear{Papadimitriou and Roughgarden}{Papadimitriou
  and Roughgarden}{2008}]%
        {papadimitriou2008computing}
\bibfield{author}{\bibinfo{person}{Christos~H Papadimitriou} {and}
  \bibinfo{person}{Tim Roughgarden}.} \bibinfo{year}{2008}\natexlab{}.
\newblock \showarticletitle{Computing correlated equilibria in multi-player
  games}.
\newblock \bibinfo{journal}{\emph{Journal of the ACM (JACM)}}
  \bibinfo{volume}{55}, \bibinfo{number}{3} (\bibinfo{year}{2008}),
  \bibinfo{pages}{1--29}.
\newblock


\bibitem[\protect\citeauthoryear{Peleg}{Peleg}{2000}]%
        {Peleg2000}
\bibfield{author}{\bibinfo{person}{David Peleg}.}
  \bibinfo{year}{2000}\natexlab{}.
\newblock \bibinfo{booktitle}{\emph{{Distributed Computing: A
  Locality-Sensitive Approach}}}.
\newblock \bibinfo{publisher}{Society for Industrial and Applied Mathematics}.
\newblock
\urldef\tempurl%
\url{https://doi.org/10.1137/1.9780898719772}
\showDOI{\tempurl}


\bibitem[\protect\citeauthoryear{Roughgarden}{Roughgarden}{[n.d.]}]%
        {roughgarden2007routing}
\bibfield{author}{\bibinfo{person}{Tim Roughgarden}.}
  \bibinfo{year}{[n.d.]}\natexlab{}.
\newblock \showarticletitle{Routing games}.
\newblock \bibinfo{journal}{\emph{Algorithmic game theory}}
  \bibinfo{volume}{18} (\bibinfo{year}{[n.\,d.]}), \bibinfo{pages}{459--484}.
\newblock


\bibitem[\protect\citeauthoryear{Schoenebeck and Vadhan}{Schoenebeck and
  Vadhan}{2012}]%
        {Schoenebeck2012}
\bibfield{author}{\bibinfo{person}{Grant~R Schoenebeck} {and}
  \bibinfo{person}{Salil Vadhan}.} \bibinfo{year}{2012}\natexlab{}.
\newblock \showarticletitle{The computational complexity of Nash equilibria in
  concisely represented games}.
\newblock \bibinfo{journal}{\emph{ACM Transactions on Computation Theory
  (TOCT)}} \bibinfo{volume}{4}, \bibinfo{number}{2} (\bibinfo{year}{2012}),
  \bibinfo{pages}{1--50}.
\newblock


\bibitem[\protect\citeauthoryear{Stewart, Mosleh, Diakonova, Arechar, Rand, and
  Plotkin}{Stewart et~al\mbox{.}}{2019}]%
        {Stewart2019}
\bibfield{author}{\bibinfo{person}{Alexander~J Stewart},
  \bibinfo{person}{Mohsen Mosleh}, \bibinfo{person}{Marina Diakonova},
  \bibinfo{person}{Antonio~A Arechar}, \bibinfo{person}{David~G Rand}, {and}
  \bibinfo{person}{Joshua~B Plotkin}.} \bibinfo{year}{2019}\natexlab{}.
\newblock \showarticletitle{Information gerrymandering and undemocratic
  decisions}.
\newblock \bibinfo{journal}{\emph{Nature}} \bibinfo{volume}{573},
  \bibinfo{number}{7772} (\bibinfo{year}{2019}), \bibinfo{pages}{117--121}.
\newblock


\bibitem[\protect\citeauthoryear{Trevisan}{Trevisan}{2012}]%
        {Trevisan12}
\bibfield{author}{\bibinfo{person}{Luca Trevisan}.}
  \bibinfo{year}{2012}\natexlab{}.
\newblock \showarticletitle{Max Cut and the Smallest Eigenvalue}.
\newblock \bibinfo{journal}{\emph{{SIAM} J. Comput.}} \bibinfo{volume}{41},
  \bibinfo{number}{6} (\bibinfo{year}{2012}), \bibinfo{pages}{1769--1786}.
\newblock
\urldef\tempurl%
\url{https://doi.org/10.1137/090773714}
\showDOI{\tempurl}


\bibitem[\protect\citeauthoryear{Vickrey and Koller}{Vickrey and
  Koller}{2002}]%
        {Vickrey2002}
\bibfield{author}{\bibinfo{person}{David Vickrey} {and} \bibinfo{person}{Daphne
  Koller}.} \bibinfo{year}{2002}\natexlab{}.
\newblock \showarticletitle{Multi-agent algorithms for solving graphical
  games}.
\newblock \bibinfo{journal}{\emph{AAAI/IAAI}}  \bibinfo{volume}{2}
  (\bibinfo{year}{2002}), \bibinfo{pages}{345--351}.
\newblock


\end{thebibliography}

\appendix

\section{Proof of Theorem~\ref{thm:local-dynamics-local}} \label{app:local-dynamics-proof}

In this section we prove Theorem~\ref{thm:local-dynamics-local}. The following lemma states more generally that the execution of best response dynamics can be simulated for any number of rounds.

\begin{lemma} \label{lem:simulation-dynamics}
	Assume that $(\mathcal{A}, u, N)$ is a graphical game and $\vec{a}$ is some strategy profile. A distributed algorithm that is given $\vec{a}$ as an input can simulate $T$ rounds of best responses, for some ordering of the play, in $O(\log^* n + T)$ rounds.
\end{lemma}

\begin{proof}
	The simulation consists of two phases. In the first phase, the nodes compute a coloring of $N^2$ (the virtual network obtained by connecting all nodes at distance at most 2 in $N$) with $k = \Delta^2 + 1$ colors. That is, each node $v$ chooses a label $c(v)$ from $\{1,2,\dots,k\}$ such that any two nodes $u$ and $v$ within distance 2 in $N$ have different labels $c(u) \neq c(v)$. This can be computed in $O(\log^* n)$ rounds~\cite{Barenboim16deterministic}
	
	In the second phase this coloring is treated as a \emph{schedule}: at round $j$ of the second phase each node with color $i = j \bmod k$ is active, applies the best response to the current strategy profile, and sends its new strategy to its neighbors. The key is that any two nodes updating their strategy at the same time do so \emph{independently}: since they are not neighbors, their choices do not depend on each other. Therefore applying best responses at all nodes of the same color class is equivalent to letting the corresponding agents play in any sequential order: given an initial strategy profile $\vec{a}$, all orderings produce the same strategy profile $a'$. Simulating all color classes one by one therefore corresponds to \emph{some} ordering of sequential play.
	
	Since there are $k = \Delta^2 + 1$ color classes, simulating one fair round of best responses takes $k$ rounds in the LOCAL model. Since we assume $\Delta$ is a constant, simulating $T$ rounds of best responses can be done in $O(T)$ rounds. With the initial coloring step we have that the total running time of the simulation is $O(\log^* n + T)$, as required.
\end{proof}

It follows that simulating best responses until convergence can be done with an additive $O(\log^* n)$ overhead. 

\begin{proof}[Proof of Theorem~\ref{thm:local-dynamics-local}]
	First, assume that the best responses start from a constant or worst-case initial strategy profile. Each node can simply choose the same initial value and simulate $T(n)$ rounds of best responses by Lemma~\ref{lem:simulation-dynamics}. Since we assume that best responses converge for any order of play in $T(n)$ rounds, it follows that $T(n)$ rounds of simulation converge as well. Computing the simulation until convergence takes $O(\log^* n + T(n))$ rounds in the deterministic LOCAL model.
	
	Next, assume that the best responses start from a random initial strategy profile. Now it is no longer possible to use deterministic algorithms to run the simulation. Using the random inputs, each node can choose a random initial strategy. Then it can simulate best responses for $T(n)$ rounds by Lemma~\ref{lem:simulation-dynamics}. Since we assume that the best responses converge with high probability and the dynamics are deterministic given the initial configuration, the simulation also converges with high probability in $T(n)$ rounds. The simulation can be computed in $O(\log^* n + T(n))$ rounds in the randomized LOCAL model.
\end{proof}

\section{Graph-theoretic analysis for Section~\ref{sec:br-inefficiency}} \label{app:graph-proofs}

\subsection{Graph constructions for the best-shot public goods game}

The \emph{domination number} $\gamma(N)$ of a graph $N$ is the size of a minimum dominating set. The \emph{independence number} $\alpha(N)$ of a graph $N$ is the size of a maximum independent set. A dominating set is \emph{perfect} if every node not in the set is adjacent to exactly one node in the set (counting the node itself). The \emph{girth} of a graph is the length of its shortest cycle.
 
The following lemma states that there exist $d$-regular graphs of logarithmic girth such that all dominating sets are large and all independent sets are small.

\begin{lemma} \label{lem:bad-graph-pgg}
	For each $d \geq 3$ and each sufficiently large $n_0$, there exists a $d$-regular graph $N$ on $n \geq n_0$ nodes with girth $g = \Omega(\log_d n)$ such that the following hold.
	\begin{enumerate}[noitemsep]
		\item The domination number $\gamma(N)$ is at least $(1+\varepsilon)(\ln d / d)n$.
		\item The independence number $\alpha(N)$ is at most $((2+\varepsilon) \ln d / d)n$.
	\end{enumerate}
\end{lemma}

\begin{proof}
	The existence of such graphs relies on properties of random $d$-regular graphs. Random $d$-regular graphs have no large independent sets and no small dominating sets with high probability. These graphs can then be modified by the standard cycle cutting technique (see e.g. \cite{bollobas04extremal,alon10constant}). The proof follows the proof of Alon~\cite[Lemma 21]{alon10constant}.
	
	A random $d$-regular graph $N$ has independence number $\alpha$ at most $((2+\varepsilon)\ln d / d)n$~\cite{FRIEZE1992} and domination number at least $((1+\varepsilon)\ln d / d)n$~\cite{Alon10star,alon10constant} with high probability. A random $d$-regular graph has only $c = O(\sqrt{n})$ cycles of length $o(\log_d n)$ in expectation~\cite{alon10constant}. Therefore it is possible to find a $d$-regular graph with all three properties.
	
	Graph $N$ might have short cycles. We will cut these one by one, starting with one of the shortest cycles. Pick an edge $e = \{u,v\}$ on that cycle. Since the graph $N$ has maximum degree $d$, there must be an edge $f = \{u', v'\}$ at distance $g$ from $e$ (we say that an edge $e$ is at distance $d$ from edge $f$ if the minimum distance between a node of $e$ and a node of $f$ is $d$). Remove $e$ and $f$ from $G$ and add edges $\{u,v'\}$ and $\{u', v\}$. Since these edges were far away, no short cycles were created.
	Continue in this way until there are no cycles of length less than $g$ left -- we are guaranteed that this process continues if we choose a suitable value for $g$.
	The graph now has girth $\Omega(\log_d n)$.
	
	Now consider any maximum independent set $I^*$ of $N'$. In the worst case each edge that was removed from $N$ goes between two nodes in $I^*$. If we remove one of each such pair of nodes from $I^*$, we obtain an independent set $I$ of $N$, the original graph. Since exactly two edges were removed for each cycle, we get that $|I| \geq |I^*| - 2c = |I^*| - O(\sqrt{n})$. By choosing a sufficiently large $n$, we have that $|I^*| \leq ((2+\varepsilon)\ln d / d)n$, as required.
	
	We can argue similarly about the domination number: if $D^*$ is a minimum dominating set of $N'$, then by adding at most $2c$ nodes to it we obtain a dominating set of $N$. Since all dominating sets of $N$ are large, $D^*$ must also be large.
\end{proof}

The second lemma states that there exist $d$-regular graphs of logarithmic girth (that is, graphs that look locally the same as graphs from Lemma~\ref{lem:bad-graph-pgg}), that are bipartite, and that have perfect dominating sets.

\begin{lemma} \label{lem:good-graph-pgg}
	For each $d \geq 3$ and every $n_0$, there exists a $d$-regular bipartite graph on $n \geq n_0$ nodes with a perfect dominating set and girth $g = \Omega(\log_d n)$.
\end{lemma}

\begin{proof}
	The proof proceeds in three steps, utilising again the cycle cutting technique.
	\begin{enumerate}[noitemsep]
		\item Construct a graph $N$ with a perfect dominating set by taking a collection of $k$ stars on $d+1$ nodes (we have $n = k(d+1)$). Let $D$ denote the set of centers of the stars and $U$ the set of leaves. Find $d-1$ disjoint perfect matchings in the complete graph on $U$, and add the corresponding edges to $N$. We call these edges \emph{leaf edges}. This is always possible for a sufficiently large $k$.
		\item Graph $N$ might have short cycles. Cut these one by one, starting with one of the shortest cycles. Pick a leaf edge $e = \{u,v\}$ on that cycle. Since the graph $N$ has maximum degree $d$, there must be a leaf edge $f = \{u', v'\}$ at distance $g$ from $e$ (we say that an edge $e$ is at distance $d$ from edge $f$ if the minimum distance between a node of $e$ and a node of $f$ is $d$). Remove $e$ and $f$ from $G$ and add edges $\{u,v'\}$ and $\{u', v\}$. Since these edges were far away, no short cycles were created. Since we crossed two leaf edges, the set $D$ is still a perfect dominating set. 
		Continue in this way until there are no cycles of length less than $g$ left -- we are guaranteed that this process continues if we choose a suitable value for $g$.
		The graph now has girth $\Omega(\log_d n)$.
		\item To construct a bipartite graph that retains the properties from the previous steps, we take a \emph{bipartite double cover} $N' = (V', E')$ of $N = (V,E)$. For each $v \in V$ we take two copies $v_1, v_2 \in V'$. Denote the first copies by $V'_1$ and the second copies by $V'_2$. For each edge $\{u,v\} \in E$ we add crossed copies $\{u_1, v_2\}$ and $\{u_2, v_1 \}$. All edges go between the sets $V'_1$ and $V'_2$, so the graph $G'$ is bipartite. By construction we do not create any new short cycles, and since $|V'| = 2|V|$ we have that the girth of $G'$ is also $\Omega(\log_d n)$. Finally, the copies of nodes corresponding to the set $D$ form a perfect dominating set in $N'$.
	\end{enumerate}
\end{proof}

\subsection{Graph constructions for the minority game}

We first construct a graph that has a cut that cuts all edges and has high girth.

\begin{lemma} \label{lem:bipartite-high-girth}
	For every $d \geq 3$ and every sufficiently large $n_0$, there exists a $d$-regular bipartite graph on $n \geq n_0$ nodes with girth $g = \Omega(\log_d n)$.
\end{lemma}

\begin{proof}[Proof of Lemma~\ref{lem:small-cuts-high-girth}]
	Friedman~\cite{Friedman03} showed that the second largest eigenvalue, in absolute value, of the adjacency matrix of a random $d$-regular graph is at most $2\sqrt{d-1}+ \varepsilon$, for any $d \geq 3$ and any $\varepsilon > 0$, with high probability. Let $\lambda_n$ denote the smallest eigenvalue. The size of the maximum cut of a graph is known to be bounded by
	\[
		\Bigl(\frac{1}{2} + \frac{1}{2d}|\lambda_n|\Bigr)|E|
	\]
	(see e.g.\ Trevisan~\cite{Trevisan12}). This implies that with high probability, the maximum cut of a random $d$-regular graph has at most a $1/2 + \sqrt{d-1}/d + \varepsilon / (2d)$-fraction of the edges.
	
	Now we apply the cycle cutting technique once again. Find a $d$-regular graph $N$ that has no large cuts and $O(\sqrt{n})$ cycles of length less than $g = \Omega(\log_d n)$. Repeatedly cut each cycle of length less than $g$ to obtain a new graph $N'$. Now consider any maximum cut $C$ of $N'$. In the worst case, all edges that were removed from $N$ were cut edges of $C$, and the new edges are no longer cut edges. Since we removed $O(\sqrt{n})$ edges, the size of the cut in $N'$ is at most $(1/2 + \sqrt{d-1}/d + \varepsilon / (2d))|E| - O(\sqrt{n})$. For sufficiently large $n$, this is at most
	\[
		(1/2 + 1/\sqrt{d-1})|E|.
	\]
\end{proof}

Next we construct another graph that is locally indistinguishable from the first graph, but has no large cuts.

\begin{lemma} \label{lem:small-cuts-high-girth}
	For every $d \geq 3$ and every sufficiently large $n_0$, there exists a $d$-regular graph on $n \geq n_0$ nodes with girth $g = \Omega(\log_d n)$ such that the maximum cut is of size at most $(1/2 + 1/\sqrt{d-1})|E|$.
\end{lemma}

\begin{proof}
	Take any bipartite $d$-regular graph on $n = 2k \geq n_0$ nodes. Then use the cycle cutting technique from Lemmas~\ref{lem:good-graph-pgg} and \ref{lem:bad-graph-pgg}. Cut cycles until the graph has girth $g = \Omega(\log_d n)$. When connecting edges make sure that they still go between the two sides of the bipartition.
\end{proof}

\subsection{Proof of Lemma~\ref{lem:best-pgg-alg}}

\begin{proof}[Proof of Lemma~\ref{lem:best-pgg-alg}]
	We will start by showing that no fast distributed algorithm can find large independent sets or small dominating sets on the networks given by Lemma~\ref{lem:good-graph-pgg}. The proof is by a standard indistinguishability argument: we show that since algorithms cannot locally distinguish between a network of size $n$ from Lemma~\ref{lem:good-graph-pgg} and a network of size $n$ from Lemma~\ref{lem:bad-graph-pgg}, it must locally behave the same on both networks.
	
	Fix $T(n)$ to be some function in $o(\log_d n)$. Let $N$ and $N'$ be sufficiently large graphs from Lemma~\ref{lem:good-graph-pgg} and Lemma~\ref{lem:bad-graph-pgg}, respectively, both with girth $g > 2T(n)+1$. Since $T(n) = o(\log_d n)$, such graphs must exist for a sufficiently large $n$. 
	
	 Both graphs are $d$-regular, and therefore any distributed algorithm running in time $T(n)$ sees locally a $d$-regular tree at every node of both graphs. That is, the graphs are locally indistinguishable from each other.

	Since $N'$ does not have any large independent sets nor small dominating sets, any algorithm must output an independent set of expected size at most $\alpha(N')$ or a dominating set of expected size at least $\gamma(N')$. Therefore every node must have the same output distribution over the assignment of the random bits under any $T(n)$-round algorithm. This implies that the algorithm must put the node into an independent set with probability at most $\alpha(N')/n$ and into a dominating set with probability at least $\gamma(N')/n$ on $N'$. Since all neighborhoods of $N$ look indistinguishable from the neighborhoods of $N'$, any algorithm with running time $T(n)$ must behave in exactly the same way on $N$ as on $N'$: the expected size of an independent set is at most $\alpha(N')$ and the expected size of a dominating set is at least $\gamma(N')$.
\end{proof}

\end{document}